\documentclass[lettersize,journal,onecolumn]{IEEEtran}
\pdfoutput=1
\usepackage[utf8]{inputenc}

\usepackage[backend=bibtex, style = numeric-comp, doi = false, url = false, isbn = false, maxbibnames = 6]{biblatex}
\bibliography{RDreferences}
\renewbibmacro{in:}{}      
\newbibmacro{string+doi}[1]{\iffieldundef{doi}{#1}{\href{https://dx.doi.org/\thefield{doi}}{#1}}}
\DeclareFieldFormat{title}{\usebibmacro{string+doi}{\mkbibemph{#1}}}
\DeclareFieldFormat[article]{title}{\usebibmacro{string+doi}{\mkbibquote{#1}}}
\DeclareFieldFormat[incollection]{title}{\usebibmacro{string+doi}{\mkbibquote{#1}}}                   
\DeclareFieldFormat[inproceedings]{title}{\usebibmacro{string+doi}{\mkbibquote{#1}}}   

\usepackage[pdfpagelabels, linktocpage=true,pdfusetitle]{hyperref}
\hypersetup{
    colorlinks=true,
    citecolor=darkblue,
    linkcolor=reddish,
    urlcolor=darkblue,
    pdfauthor={},
    pdfsubject={}
}
\usepackage{tikz,pgfplots}
\definecolor{darkblue}{rgb}{0.15,0.35,0.55}
\definecolor{reddish}{rgb}{.8, 0.2, 0.2}

\usepackage[margin=1em]{caption}
\usepackage{subcaption}

\usepackage{amssymb,amsfonts,amsmath}
\usepackage{graphicx}
\usepackage{indentfirst}
\usepackage{dsfont}

\allowdisplaybreaks[4]

\usepackage{mathtools}
\usepackage{enumerate}
\usepackage{url}

\usepackage{amsthm}
\usepackage{mathtools}
\usepackage{mathrsfs}
\usepackage{braket}
\usepackage{physics}

\theoremstyle{plain}

\newtheorem{theorem}{Theorem}
\newtheorem{corollary}[theorem]{Corollary}

\newtheorem{lemma}[theorem]{Lemma}
 
\newtheorem{definition}[theorem]{Definition}  
\newtheorem{remark}[theorem]{Remark}  
\newtheorem*{remark*}{Remark}   
\renewcommand\qedsymbol{$\blacksquare$}
\newenvironment{proof-of}[1][{\hspace{-\blank}}]{{\medskip\noindent\textit{Proof~{#1}.\ }}}{\hfill\qedsymbol}

\renewcommand{\Tr}{{\operatorname{Tr}\,}}

\renewcommand{\dim}{{\operatorname{dim}}}
\newcommand{\cid}{{\cal I}}
\newcommand{\proj}[1]{|#1\rangle\!\langle #1|}

\newcommand{\KI}{{\mathrm{KI}}}

\newcommand{\nc}{\newcommand}
\nc{\rnc}{\renewcommand}
\nc{\avg}[1]{\langle#1\rangle}
\nc{\Rank}{\operatorname{Rank}}
\nc{\smfrac}[2]{\mbox{$\frac{#1}{#2}$}}
\renewcommand{\tr}{\operatorname{Tr}}
\nc{\ox}{\otimes}
\nc{\dg}{\dagger}
\nc{\dn}{\downarrow}
\nc{\cA}{{\cal A}}
\nc{\cB}{{\cal B}}
\nc{\cC}{{\cal C}}
\nc{\cF}{{\cal F}}
\nc{\cG}{{\cal G}}
\nc{\cH}{{\cal H}}
\nc{\cI}{{\cal I}}
\nc{\cJ}{{\cal J}}
\nc{\cK}{{\cal K}}
\nc{\cL}{{\cal L}}
\nc{\cM}{{\cal M}}
\nc{\cN}{{\cal N}}
\nc{\cO}{{\cal O}}
\nc{\cP}{{\cal P}}
\nc{\cQ}{{\cal Q}}
\nc{\cR}{{\cal R}}
\nc{\cS}{{\cal S}}
\nc{\cX}{{\cal X}}
\nc{\cY}{{\cal Y}}
\nc{\cZ}{{\cal Z}}
\nc{\csupp}{{\operatorname{csupp}}}
\nc{\qsupp}{{\operatorname{qsupp}}}
\nc{\rar}{\rightarrow}
\nc{\lrar}{\longrightarrow}
\nc{\polylog}{{\operatorname{polylog}}}
\nc{\1}{\mathds{1}}
\nc{\wt}{{\operatorname{wt}}}

\nc{\RR}{{{\mathbb R}}}
\nc{\CC}{{{\mathbb C}}}
\nc{\FF}{{{\mathbb F}}}
\nc{\NN}{{{\mathbb N}}}
\nc{\ZZ}{{{\mathbb Z}}}
\nc{\PP}{{{\mathbb P}}}
\nc{\QQ}{{{\mathbb Q}}}
\nc{\UU}{{{\mathbb U}}}
\nc{\EE}{{{\mathbb E}}}

\nc{\Hom}[2]{\mbox{Hom}(\CC^{#1},\CC^{#2})}
\nc{\rU}{\mbox{U}}

\nc{\ob}[1]{#1}

\def\R{\mathscr{R}}

\title{Rate-Distortion Theory for Mixed States}

\author{
Zahra~Baghali~Khanian, Kohdai~Kuroiwa, Debbie~Leung
\thanks{Z.~B.~Khanian is with Munich Center for Quantum Science and Technology \& Zentrum Mathematik, Technical University of Munich. E-mail: \href{mailto:zbkhanian@gmail.com}{\texttt{zbkhanian@gmail.com}}}
\thanks{K.~Kuroiwa is with the Institute for Quantum Computing \& Department of Combinatorics and Optimization, University of Waterloo and the Perimeter Institute for Theoretical Physics. E-mail: \href{mailto:kkuroiwa@uwaterloo.ca}{\texttt{kkuroiwa@uwaterloo.ca}}}
\thanks{D.~Leung is with the Institute for Quantum Computing \& Department of Combinatorics and Optimization, University of Waterloo and the Perimeter Institute for Theoretical Physics. E-mail: \href{mailto:wcleung@uwaterloo.ca}{\texttt{wcleung@uwaterloo.ca}}}
\thanks{A part of the results in this paper was presented at the 2023 IEEE International Symposium on Information Theory (ISIT), held from June 25 to 30, 2023 in Taipei, Taiwan.}
}%

\begin{document}
\maketitle 

\begin{abstract}
This paper is concerned with quantum data compression of asymptotically many
independent and identically distributed copies of ensembles of mixed quantum 
states.  The encoder has access to a side information system. 
The figure of merit is per-copy or local error criterion.
Rate-distortion theory studies the trade-off between the compression 
rate and the per-copy error. 
The optimal trade-off can be characterized by the rate-distortion
function, which is the best rate given a certain distortion.  
In this paper, we derive the rate-distortion function of mixed-state 
compression.  
The rate-distortion functions in the entanglement-assisted and
unassisted scenarios are in terms of a single-letter mutual information
quantity and the regularized entanglement of purification,
respectively.
For the general setting where the consumption of both 
communication and entanglement are considered, we present the full 
qubit-entanglement rate region. 
Our compression scheme covers both
blind and visible compression models (and other models in between)
depending on the structure of the side information system.
\end{abstract}

\tableofcontents

\section{Introduction}

\textit{Data compression} is concerned with the task of transmitting
data with high fidelity while minimizing the communication cost.  The
cost reduction is made possible by the structure and redundancy in the
data, which is often modeled as multiple samples of a (data) source.
The compression rate is the communication cost per sample.  
\textit{Quantum data compression} was pioneered by Schumacher in
\cite{Schumacher1995}.  He defined two notions of a quantum source.
The first notion of a quantum source is a quantum system together with
correlations with a purifying \textit{reference} system.  The second
notion of a quantum source is called an ensemble of pure states where
a pure quantum state is drawn according to some known probability
distribution.
Both notions are ``pure state sources.'' 
Subsequently, the notion of ensemble sources was generalized to
ensembles of mixed states in
\cite{Horodecki1998,Horodecki2000,Koashi2001,Koashi2002}.
An ensemble source is equivalently defined as a classical-quantum state where
the classical system plays the role of an inaccessible reference
system.
All these seemingly distinct definitions of a quantum source were
unified and generalized 
in \cite{ZBK_PhD,general_mixed_state_compression,ZK_mixed_state_ISIT_2020}
where a general quantum source is defined as a quantum system together
with correlations with a most general reference system (which is not
necessarily a purifying system nor a classical system).  
We note that the reference system plays a crucial role in defining the
correctness or fidelity of the transmission; it has to preserve the
correlations between the transmitted data and the reference.  
Without considering such correlations, any known classical or quantum
system can be produced locally by the receiver without any input from
the sender.
The source specifies what correlations must be preserved in the
compression task, as well as what data is to be transmitted.

Data compression, which achieves high fidelity transmission, can be
extended to a \emph{rate-distortion} model which allows for a
potential decrease of compression rate at the expense of having a
larger error \cite{shannon1959,Berger1975}.
In the quantum setting, Schumacher's quantum data compression model
was generalized by Barnum \cite{Barnum2000} to a rate-distortion model
for pure state sources.
Even though both Schumacher and Barnum considered asymptotically many
i.i.d.$\,$(independent identically distributed) copies of a quantum
source, their set-ups involve different error definitions.
More specifically, Schumacher considered the error for the whole block
of data (called block or global error) whereas Barnum considered the
error for each copy of the source (called per-copy or local error).
Under the global error definition in Schumacher's quantum data
compression model for pure state sources,
the strong converse theorems state that any reduction of rate below
the optimal value causes the fidelity to vanish for asymptotically
large blocks of data
\cite{Richard1994,Winter1999}.  
Recently, strong converse for mixed state sources was also studied
\cite{Khanian2022}.
In contrast, under the local error definition in Barnum's 
quantum rate-distortion model, 
compression rate trades more smoothly with error. 
Considering both global and local error criteria, quantum data
compression and rate-distortion have been extensively investigated
\cite{Schumacher1995,Richard1994,Barnum1996,
  Lo1995,Horodecki1998,Barnum2001,Dur2001,Kramer2001,Bennett2002,
  Winter2002, Horodecki2000,Hayashi2006,Koashi2001,
  Bennett2001,Bennett2005,Abeyesinghe2006,Jain2002,Bennett2014a,
  Barnum2000, Devetak2002, Datta2013b, Datta2013c,
  ZBK_PhD,ZK_cqSW_2018,general_mixed_state_compression,
  Khanian2022,Khanian2021,ZK_mixed_state_ISIT_2020,
  ZK_QSR_ensemble_ISIT_2020,ZK_cqSW_ISIT_2019,
  ZK_Eassisted_ISIT_2019,Anshu2019,Kuroiwa2022}.

The ultimate goal of quantum rate distortion theory is to elucidate
the optimal trade-off between the compression rate and the allowed 
local error.  
This can be characterized by the \textit{rate-distortion function},
which is the lowest rate achieved by a protocol satisfying an allowed
distortion of the data.
For pure-state sources, 
Barnum derived a lower bound on the rate-distortion function
\cite{Barnum2000}.
Subsequent work in references \cite{Datta2013b,Datta2013c}
characterized the rate-distortion function of pure-state sources in
various scenarios with and without entanglement assistance and side
information.
They found that the rate-distortion functions of entanglement assisted
compression and unassisted compression are characterized by the
quantum mutual information and the entanglement of purification
between the decoder's systems and the reference systems, respectively.
(See Section \ref{sec:notation} for the definitions of these correlation
measures.)  
Moreover, for pure-state sources where both the encoder and the
decoder have access to a side information system, the rate-distortion
function is characterized in
\cite{ZK_QSR_ensemble_ISIT_2020,Khanian2021}.
However, quantum rate-distortion theory in the case of
\textit{mixed-state sources} remains largely unexplored. 
In the special case for the \textit{blind compression} of ensembles of
mixed states, \cite{Koashi2001} showed that the optimal rate is the
same for vanishing local and global error; but the rate can be 
sensitive to a small allowed finite local error with potentially 
large rate reduction~\cite{Anshu2019,Kuroiwa2022}.

In this paper, we investigate rate distortion coding for mixed-state
sources, considering both scenarios with and without assistance of
entanglement.
We assume a mixed-state quantum \emph{ensemble} source with a side
information system available to the encoder, which covers both of
\textit{blind compression} and \textit{visible compression} as special
cases.
First, we consider the case in which we have free entanglement. 
We derive the rate-distortion function for this case, which is an
optimized expression of quantum mutual information between the decoder's 
system and a composite system consisting of the classical reference
system and other systems that purify the source
(Theorem~\ref{thm:assisted}).  Surprisingly, the rate-distortion
function can be given in a single-letter form using quantum mutual
information.
Then, we explore the case without any entanglement assistance.
In this case, we prove that the rate-distortion function is given by
the \textit{entanglement of purification} between the decoder's
systems and the same composite system as in the entanglement-assisted
case (Theorem~\ref{thm:unassisted}).
We note that this expression for the rate-distortion function is 
a multi-letter formula.  
So far, our results are derived for the most general side information.
To simplify the expression for the unassisted case, we apply it
to the two special scenarios of visible and blind compressions.
For visible compression, the expression simplifies with fewer
systems involved (Corollary~\ref{cor:unassisted_visible}).  
For blind compression with vanishing local error, we show that the
rate-distortion function becomes a single-letter formula involving the
entanglement of purification (Theorem~\ref{thm:blind_unassisted}).
As a corollary, we obtain an information theoretic identity by
equating this expression to the optimal blind compression rate
obtained in \cite{Koashi2001}.
As a final result on deriving rate-distortion functions for the most
general side information, and generalizing both entanglement-assisted
and unassisted cases, we consider general qubit-entanglement rate
region.
We derive the necessary and sufficient condition for the qubit and
entanglement rates with which rate-distortion compression is
achievable (Theorem~\ref{thm:full_rate}).

Our work contributes to a rich body of quantum rate-distortion theory
by showing the optimal trade-offs of mixed-state rate-distortion
compression with and without entanglement, which have not been
explored before.
Our achievability protocols share similarities with those
in \cite{Khanian2021} but the source, the rate expressions, and the converse
proofs differ significantly.  
Our results also have interesting implications for visible and blind compression by considering the previously known optimal compression rates for these cases. 
The rest of this paper is organized as follows. 
We review basic concepts and notations used in this paper in Sec.~\ref{sec:notation}. 
Then, we introduce our setup of rate distortion coding for mixed state ensemble sources in Sec.~\ref{sec:setup}. 
Finally, we show our main results in Sec.~\ref{sec:results}. 
We first show the optimal rate of entanglement assisted rate distortion coding in Sec.~\ref{subsec:assisted}, and then we analyze the optimal rate of unassisted case in Sec.~\ref{subsec:unassisted}.  
Moreover, we derive the full rate region of rate distortion coding for ensemble sources in Sec.~\ref{subsec:rate_region}. 
We conclude the paper with a discussion in Sec.~\ref{sec:discon}, further comparing our work with previous
results and considering additional applications.

\section{Background and Notations}~\label{sec:notation}
In this section, we review background materials and introduce notations for the paper.
Throughout, we use capital letters 
$A,B,\ldots$ to represent quantum systems. 
For a quantum system $A$, we let $\mathcal{H}_A$ denote the corresponding complex Hilbert space. 
We only consider finite-dimensional quantum systems in this paper. 
Given a quantum system $A$, we use $\mathrm{dim}(A)$ to denote the dimension of the corresponding Hilbert space.

Given two quantum systems $A$ and $B$ with $\mathrm{dim}(A) \leq
\mathrm{dim}(B)$, we write an isometry from $A$ to $B$ as $U^{A\to
  B}$.  Whenever $A=B$ we abbreviate $A\to B$ by $A$ alone, for
example, $\1^A$ denotes the identity operator on $A$.
    For a linear operator $Y$ on system $A$,
    let $s_1(Y), \cdots s_{\mathrm{dim}(A)}(Y)$ denote the singular values of $Y$.
    Then, the \textit{trace norm} of $Y$ is defined as  
        $\|Y\|_1 = \sum_{i=1}^{\mathrm{dim}(A)} s_i(Y)$. 
For systems $A$ and $B$, we write a quantum channel from $A$ to $B$,
that is, a completely positive and trace-preserving linear map from
square matrices on $\mathcal{H}_A$ to those on $\mathcal{H}_B$ as
$\mathcal{N}^{A\to B}$. (A quantum channel is also called a CPTP map.)
We also write the Stinespring dilation \cite{Stinespring1955} of the
channel $\mathcal{N}^{A \to B}$ as $U^{A\to BE}_{\mathcal{N}}$ with
$E$ being the corresponding environment system.  Similar to the
notation for operators, we only list the common space when the input
and output spaces are the same, for example, $\cid^A$ denotes the
identity channel on $A$.

A quantum ensemble is defined as follows.
\begin{definition}~\label{def:ensemble}
    Let $A$ be a quantum system, and let $\Sigma$ be an alphabet. 
    A quantum ensemble is a set of pairs of a positive real number and a quantum state
    \begin{equation*}
        \{p_x,\rho^A_x\}_{x\in\Sigma},
    \end{equation*}
    where $\{p_x:x\in\Sigma\}$ forms a probability distribution; that is,
    $0\leq p_x \leq 1$ 
    for all $x\in\Sigma$ and
    $\sum_{x\in\Sigma} p_x = 1$.
    In addition, given a quantum ensemble $\{p_x,\rho^A_x\}_{x\in \Sigma}$, 
we use a classical-quantum state  
    \begin{equation*}
        \rho^{AX} \coloneqq \sum_{x\in\Sigma} p_x \rho^A_x \otimes \proj{x}^X 
    \end{equation*}
    to represent this ensemble, 
    where $X$ is a classical reference system with orthonormal basis $\{\ket{x}\}_{x\in \Sigma}$.  
    
\end{definition}

Given a quantum ensemble, we can consider a decomposition of each
state into a \textit{classical part}, a \textit{quantum part}, and a
\textit{redundant part}, which is known as \textit{Koashi-Imoto (KI)
  decomposition}~\cite{Koashi2002} named after the authors.


\begin{theorem}[\cite{Koashi2002}] \label{thm:ki1}
Let $\rho^{AX}$ be a classical-quantum state associated with a quantum ensemble. 
Then, there exist a joint quantum system $CNQ$ and a corresponding isometry
$U_{\KI}: A \hookrightarrow CNQ$ satisfying the following conditions.
\begin{enumerate}[(i)]
\item The state $\omega^{C N
  Q X} = (U_{\KI} \otimes \1^X)\rho^{AX} (U_{\KI}^{\dagger} \otimes \1^X)$ can be expressed as
    \begin{align}
    \label{KI_item1}
    \omega^{C N Q X} =
    \sum_x p_x \sum_c p_{c|x} \proj{c}^{C} \otimes \omega_c^{N} \otimes \rho_{cx}^{Q} 
        \otimes \proj{x}^X
    \end{align}
    where the set of vectors $\{ \ket{c}^{C}\}_{c \in \Xi}$ form an orthonormal
    set for system $C$, and for each $x$, $p_{c|x}$ is a
    distribution over $c$ conditioned on $x$.  The states
    $\omega_c^{N}$ and $\rho_{cx}^{Q}$ live in systems $N$ and $Q$,
    respectively. 

  \item
      Consider an arbitrary $c \in \Xi$.  
      Let $Q_c$ be the smallest subspace of $Q$ that contains the 
      support of $\rho^Q_{cx}$ for each $x \in \Sigma$.
      If a projector $P^{Q_c}$ on $Q_c$ satisfies  
    \begin{equation}
        P^{Q_c}  (p_{c|x} \rho^{Q}_{cx} )  
      = ( p_{c|x} \rho^{Q}_{cx} )  P^{Q_c} , 
    \end{equation}
      for all $x$, then, $P^{Q_c} = \1^{Q_c}$ or $P^{Q_c} = 0$. 
    
    \item
      Consider an arbitrary pair of distinct $c,c' \in \Xi$ and  
      any collection of positive numbers $\{\alpha^{(c,c')}_x\}_{x\in\Sigma}$.
      There is no unitary operator $V^{Q}$ satisfying 
    \begin{equation}
        V^{Q} (p_{c|x} \rho^{Q}_{cx})  
      = \alpha^{(c,c')}_x (p_{c'|x} \rho^{Q}_{c'x}) V^{Q}   
    \end{equation}
simultaneously for all $x \in \Sigma$.
\end{enumerate} 
\end{theorem}

Henceforth, we call the isometry $U_{\KI}$ and the state $\omega^{C N
  Q X}$ in Theorem \ref{thm:ki1} 
the Koashi-Imoto (KI)
isometry and KI-decomposition of the state $\rho^{AX}$, respectively.
In the KI decomposition, all states in a given ensemble are written in
a block-diagonal form with the same block structure.  Register $C$
contains the label for the block; thus $C$ is called the classical
part of the ensemble.  The state in register $N$ depends only on $c$,
and it can be recovered using $c$ without the state label $x$; in this
sense, $N$ is called the redundant part of this ensemble.  On the
other hand, the state in register $Q$ depends on both $x$ and $c$, and
$Q$ is called the (non-redundant) quantum part of the ensemble.
Items (ii) and (iii) in Theorem \ref{thm:ki1} state
  that the decomposition is maximal (see \cite{Koashi2002} for
  definition).  Intuitively, item (ii) prevents further decomposition
  of each block, while item (iii) prohibits the non-redundant parts of
  one block to be related to that of another, which may generate
  larger redundant parts.  Together, this maximality prevents more
  redundant parts to be created from the non-redundant parts, a
  condition crucial for evaluating the optimal compression rate.

Since $U_{\KI}$ can be applied by the encoder and be reversed
by the decoder without affecting the compression task, so, without
loss of generality, we assume $\rho^{AX} = \omega^{C N Q X}$ for the
rest of the paper.  

We can also determine the form of quantum channels that preserve a given quantum ensemble using its KI decomposition. 
\begin{lemma}[\cite{Koashi2002}]\label{lem:kiinvar}
    Let $\rho^{AX}$ be a classical-quantum state with KI decomposition 
    \begin{equation*}
       \omega^{CNQX} = (U_{\KI}\!\otimes\! \1_X)\rho^{AX} (U_{\KI}^{\dagger} \otimes \1_X)
         =  \sum_x p_x \sum_c p_{c|x} \proj{c}^{C} \otimes \omega_c^{N} \otimes \rho_{cx}^{Q} 
        \otimes \proj{x}^X. 
    \end{equation*}
    Let $\Lambda^{A}:A\to A$ be a quantum channel that preserves the state $\rho^{AX}$; that is, $\left(\Lambda^{A}\otimes\cid^{X}\right)\left(\rho^{AX}\right) = \rho^{AX}$.
    Let $U^{A\to AE}_{\Lambda}$ be any Stinespring dilation of $\Lambda^{A}$ with environment system $E$. 
    Then, $U^{A\to AE}_{\Lambda}$ has the form 
    \begin{equation}
    \label{eq:KI_unitary}
       (U_{\KI} \otimes \1^E) \, U^{A\to AE}_{\Lambda} (U_{\KI}^{\dagger})
        = \sum_c \proj{c}^{C} \otimes U_c^{N\to NE} \otimes \1^{Q}, 
    \end{equation}
    where for each $c$, $U_c:N \hookrightarrow N E$ is an isometry satisfying 
    $\Tr_{\!E} \; [U_c \, \omega_c \, U_c^{\dagger}]=\omega_c$. 
\end{lemma}

Given a quantum ensemble, we can freely remove and attach the redundant part as shown in the following lemma. 
\begin{lemma}[KI operations~\cite{Koashi2002}] \label{lem:kiop}
  Let $\rho^{AX}$ be a classical-quantum state and $\omega^{CQNX}$ be its KI decomposition
  as given in Theorem \ref{thm:ki1}. 
  Then, there are 
  quantum channels $(\mathcal{K}_{\mathrm{off}}^{A\to CQ},\mathcal{K}_{\mathrm{on}}^{CQ\to A})$ such that 
    \begin{align*}
        \mathcal{K}_{\mathrm{off}}^{A\to CQ}(\rho^{AX}) &= \omega^{CQX}, \\
        \mathcal{K}_{\mathrm{on}}^{CQ\to A}(\omega^{CQX}) &= \rho^{AX}. 
    \end{align*} 
\end{lemma}

For a quantum state $\rho^A$ on system $A$, the von Neumann entropy of $\rho^A$ is defined as 
\begin{equation*}
    S(\rho^A) \coloneqq S(A)_{\rho} \coloneqq -\tr(\rho^A\log \rho^A). 
\end{equation*}
where $\log$ is taken base $2$ throughout the paper.  
The von Neumann entropy has a dimension bound given by 
\begin{equation*}
    S(\rho^A) \leq \log \mathrm{dim}(A).  
\end{equation*}
The von Neumann entropy is subadditive as shown in the following lemma. 
\begin{lemma}[\cite{Araki-Lieb-70}]~\label{lem:subadditivity}
    Let $\rho^{AR}$ be a quantum state on a composite system $AR$. 
    Then, 
    \begin{equation*}
        S(AR)_{\rho} \leq S(A)_{\rho} + S(R)_{\rho}. 
    \end{equation*}
\end{lemma}
\noindent The von Neumann entropy is asymptotically continuous.
\begin{lemma}[Fannes Inequality \cite{Fannes1973}]~\label{lem:fannes}
Let $\rho$ and $\sigma$ be density matrices on a Hilbert space
of dimension $d$.  If $\frac{1}{2} \|\rho-\sigma\|_1 \leq \epsilon \leq 1-\frac{1}{d}$, 
then, $|S(\rho)-S(\sigma)| \leq \epsilon \log(d-1) +h_2(\epsilon)$ 
where $h_2$ is the binary entropy function. 
\end{lemma}

For a quantum state $\rho^{AR}$ on a composite system $AR$, the quantum mutual information of $\rho^{AR}$ is defined as 
\begin{equation*}
    I(A:R)_{\rho} \coloneqq S(A)_{\rho} + S(R)_{\rho} - S(AR)_{\rho}.
\end{equation*}
The quantum mutual information also has a dimension bound: 
\begin{equation*}
    I(A:R)_{\rho} \leq 2 \log \left( \; \mathrm{min}(\mathrm{dim}(A),\mathrm{dim}(R)) \; \right). 
\end{equation*}
We will use the following property of the quantum mutual information:
\begin{lemma}[\cite{Datta2013b}]~\label{lem:convexity}
    Let $\rho^{AR}$ be a quantum state on a composite system $AR$. 
    Let $\mathcal{N}_1$ and $\mathcal{N}_2$ be quantum channels from system $A$ to system $B$, 
    and let $\lambda \in (0,1)$. 
    Then, 
    \begin{equation*}
        I(B:R)_{\lambda \left(\mathcal{N}_1\otimes \cid \right)(\rho) + (1-\lambda)\left(\mathcal{N}_2\otimes \cid \right)(\rho)} \leq \lambda I(B:R)_{\left(\mathcal{N}_1\otimes \cid \right)(\rho)} + (1-\lambda) I(B:R)_{\left(\mathcal{N}_2\otimes \cid \right)(\rho)}. 
    \end{equation*}
\end{lemma}
\noindent In addition, the quantum mutual information follows  
the data-processing inequality:
\begin{lemma}~\label{lem:data_processing}
    Let $\rho^{AR}$ be a quantum state on a composite system $AR$, 
    and let $\mathcal{N}^{A\to B}$ and $\mathcal{M}^{R\to R'}$ be quantum channels. 
    Then, 
    \begin{align*}
        I(A:R)_{\rho} &\geq I(B:R')_{(\mathcal{N}\otimes\mathcal{M})(\rho)}. 
    \end{align*}
\end{lemma}
\noindent The quantum mutual information also satisfies a property called \textit{superadditivity}. 
\begin{lemma}[\cite{Datta2013b}]~\label{lem:superadditivity-orig}
    Let $\rho^{A_1R_1}$ and $\sigma^{A_2R_2}$ be pure quantum states on composite systems $A_1R_1$ and $A_2R_2$. 
    Let $\mathcal{N}^{A_1A_2\to B_1B_2}$ be a quantum channel, and 
    $\omega^{B_1B_2R_1R_2} \coloneqq \mathcal{N}^{A_1A_2\to B_1B_2}(\rho^{A_1R_1}\otimes\sigma^{A_2R_2})$.
    Then, 
    \begin{align*}
        I(B_1B_2:R_1R_2)_{\omega} \geq I(B_1:R_1)_{\omega} + I(B_2:R_2)_{\omega}. 
    \end{align*}
\end{lemma}
\noindent 
We will use the following generalization of the above lemma.
\begin{lemma}~\label{lem:superadditivity}
    For any integer $n \geq 2$, let $\rho_1^{A_1R_1},\cdots,\rho_n^{A_nR_n}$ be pure quantum states on $A_1R_1,\cdots,A_nR_n$. 
    Let $\mathcal{N}^{A_1\cdots A_n\to B_1 \cdots B_n}$ be a quantum channel, and 
    $\omega^{B_1 \cdots B_n R_1 \cdots R_n} \coloneqq \mathcal{N}^{A_1\cdots A_n\to B_1 \cdots B_n}(\rho_1^{A_1R_1}\otimes\cdots \otimes \rho_n^{A_nR_n})$.
    Then, 
    \begin{align*}
        I(B_1 \cdots B_n:R_1 \cdots R_n)_{\omega} \geq I(B_1:R_1)_{\omega} + \cdots + I(B_n:R_n)_{\omega}. 
    \end{align*}
\end{lemma}
\begin{proof}
We sketch a proof by induction on $n$.  Lemma \ref{lem:superadditivity-orig} is the base case for $n=2$. 
Furthermore, using Lemma \ref{lem:superadditivity-orig},  
$$I(B_1 \cdots B_n:R_1 \cdots R_n)_{\omega} \geq I(B_1\cdots B_{n-1}:R_1\cdots R_{n-1})_{\omega} + I(B_n:R_n)_{\omega}.$$  
Let $\tilde{\omega} = \tr_{B_n R_n} \omega$.  We can see that $\tilde{\omega} = \tilde{\mathcal{N}}^{A_1\cdots A_{n-1}\to B_1 \cdots B_{n-1}}(\rho_1^{A_1R_1}\otimes\cdots\otimes \rho_{n-1}^{A_{n-1}R_{n-1}})$ for some quantum channel $\tilde{\mathcal{N}}$, which acts by attaching $\rho_n^{A_nR_n}$, applying $\mathcal{N}$, and then tracing out $B_n R_n$.  Thus we can apply the induction hypothesis to $I(B_1 \cdots B_{n-1}:R_1 \cdots R_{n-1})_{\tilde{\omega}}$.  Finally, noting that $I(B_i:R_i)_{\omega} = I(B_i:R_i)_{\tilde{\omega}}$ completes the induction step.  
\end{proof}

\noindent 
To characterize the rate-distortion function, we use the entanglement of purification \cite{Terhal2002}
which can be expressed as follows. 
\begin{theorem}[\cite{Terhal2002}]\label{thm:entanglement_of_purification}
The entanglement of purification of a bipartite state $\rho^{AR}$ is given by 
\begin{align}
    E_{p}(A:R)_{\rho}
= \min_{\cN^{A'\to A''}} S(AA'')_{\sigma}, \nonumber
\end{align}
where $\rho^{AR}$ is first purified to $\ket{\psi}^{AA'R}$ and then 
a quantum channel $\cN^{A'\to A''}$ is applied to the purifying system $A'$ 
to minimize the entropy of $\sigma^{AA''}=(\cid^A \otimes \cN^{A'\to A''})(\rho^{AA'})$.
\end{theorem}
\noindent The entanglement of purification is upper-bounded by the von Neumann entropy.  
\begin{lemma}[\cite{Terhal2002}]~\label{lem:upper_bound}
    Let $\rho^{AR}$ be a quantum state on a composite system $AR$. 
    Then, 
    \begin{align*}
        E_p(A:R)_{\rho} \leq \min\{S(A)_{\rho},S(R)_{\rho}\}.
    \end{align*}
\end{lemma}
\noindent The entanglement of purification also satisfies the monotonicity property, 
which is analoguous to the data processing inequality for quantum mutual information.
\begin{lemma}[\cite{Terhal2002}]~\label{lem:monotonicity}
    Let $\rho^{AR}$ be a quantum state on a composite system $AR$, 
    and let $\mathcal{N}^{A\to B}$ be a quantum channel. 
    Then, 
    \begin{align*}
        E_p(A:R)_{\rho} &\geq E_p(B:R)_{(\mathcal{N}\otimes \; \cid)(\rho)}.
    \end{align*}
\end{lemma}

\noindent 
Our results rely on quantum state redistribution~\cite{Devetak2008a,Yard2009},
which can be summarized as follows.  
\begin{theorem}[Quantum state redistribution~\cite{Devetak2008a,Yard2009}]\label{qsrrecap}
Consider an arbitrary tripartite state on $ACB$, with purification
$|\psi\rangle^{ACBR}$.  Consider $n$ copies of the state for large
$n$, on systems $A_1, \cdots A_n, C_1, \cdots C_n, B_1, \cdots B_n,
R_1, \cdots R_n$.  Suppose initially Alice has systems $A_1, \cdots
A_n, C_1, \cdots C_n$, and Bob has systems $B_1, \cdots B_n$.  Then,
there is a protocol transmitting
$nQ$ qubits 
from Alice to Bob for
$Q=\frac{1}{2} I(C:R|B) + \eta_n$, and consuming $nE$ ebits shared between them, 
where $Q+E = S(C|B) + \eta_n$, so that the final state is
$\epsilon_n$-close to $(|\psi\rangle^{ACBR})^{\otimes n}$
but now $C_1, \cdots C_n$ are in the possession
of Bob, and such that
$\{\eta_n\}$, $\{\epsilon_n\}$ are vanishing non-negative sequences.
\end{theorem}

\noindent In our paper, we only consider the situation when Alice and
Bob share ebits, and $B$ is trivial, so, there is no conditioning on
$B$ in the rates.

\section{Setup for Rate Distortion Coding for Ensemble Sources}
\label{sec:setup}
In this section, we define our setup for rate distortion coding. 
This is summarized in Fig.~\ref{fig:compression model} below, 
and described in detail in the rest of the section.  
\begin{figure}[ht] 
\centering
  \includegraphics[width=0.7\textwidth]{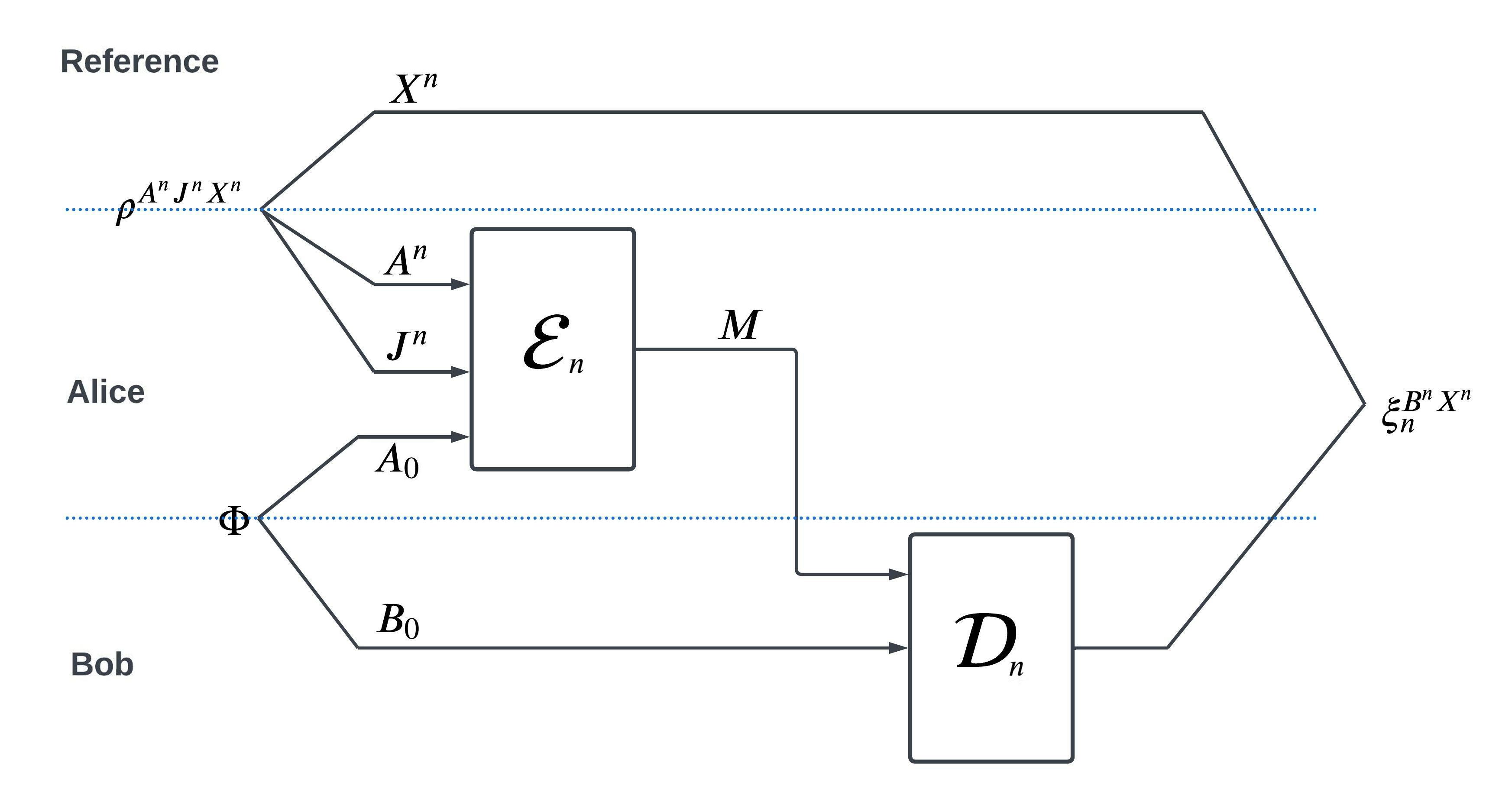} 
  \caption{Circuit diagram of the compression model. Dotted lines are used to demarcate domains controlled
by the different participants. }
  \label{fig:compression model}
\end{figure}

We consider a source given by the ensemble $\{p_x,\rho_x^A\otimes
\proj{j_x}^J\}_{x\in\Sigma}$ with side information in the system $J$.
Note that the side information can be quantum, as the states
$\proj{j_x}$ need not be orthonormal.  
The most general side information need not be pure,
and can be entangled with the system $A$.
We consider the simpler pure state side information when deriving
most of our results; in Section \ref{sec:discon}, we formalize 
the most general side information, and outline the extensions 
of our results to the general setting.  
When $J$ is one-dimensional, there is no side information  
and the task is called blind compression.  In contrast, when $|j_x\rangle$ 
contains a classical description of the state $\rho_x$, meaning that the 
sender knows the state to be compressed, then, the task 
is called visible compression.  Our general model for side information 
includes both blind and visible compression as special cases. 
The ensemble source can be equivalently defined as 
a classical-quantum (cq) state where the system $X$ plays the role of a
classical 
reference system 
\begin{equation}
    \rho^{AJX} \coloneqq \sum_{x \in \Sigma} p_x \rho_x^A \otimes\proj{j_x}^{J} \otimes \proj{x}^X , 
\label{eq:source}
\end{equation}
and symbols in $\Sigma$ are associated with an orthonormal
basis $\{\proj{x}\}_{x \in \Sigma}$ on $X$. 
We also consider a purification of the source: 
\begin{equation}
    \ket{\psi}^{AJXX'R} \coloneqq \sum_{x \in \Sigma} \sqrt{p_x} \ket{\varphi_x}^{AR} \ket{j_x}^J \ket{x}^X \ket{x}^{X'}, 
\label{eq:pursource}
\end{equation}
where $R$ purifies $\rho_x^A$ and $X'$ purifies the distribution over the reference system $X$.  
In rate distortion coding, 
the source first generates $\ket{\psi^n}^{A^nJ^nX^nX'^nR^n}$, which contains  
$n$ copies of $\ket{\psi}^{AJXX'R}$.
We call the sender or the encoder Alice, and the receiver
or the decoder Bob. 
We suppose that Alice and Bob initially share some entangled state $\ket{\Phi}$
in systems $A_0 B_0$.
This assumption is fully general, as the dimension of
$A_0 B_0$ can be chosen to be $1$ for the unassisted case, unbounded for the
entanglement-assisted case, and anything in between when we study the
entanglement-qubit rate pair.  
Alice applies an encoding channel $\mathcal{E}_n^{A^nJ^nA_0 \to M}$, 
and sends system $M$ to Bob. 
Receiving $M$, Bob applies a decoding channel $\mathcal{D}_n^{MB_0\to B^n}$. 
We define 
\begin{align*}
    \sigma_n^{MB_0X^n} &\coloneqq (\mathcal{E}_n^{A^nJ^nA_0 \to M}\otimes \cid^{X^nB_0})(\rho^{A^nJ^nX^n}\otimes \proj{\Phi}^{A_0B_0}),\\ 
    \xi_n^{B^nX^n} &\coloneqq (\mathcal{D}_n^{MB_0\to B^n}\otimes \cid^{X^n})(\sigma_n^{MB_0X^n}).
\end{align*}
Furthermore, consider the Stinespring dilations $U^{A^nJ^nA_0 \to MW_A}_{\mathcal{E}_n}$ and $U^{MB_0\to B^nW_B}_{\mathcal{D}_n}$ for the encoding and the decoding maps. 
We consider the following purifications for the states $\sigma_n$ and $\xi_n$, 
\begin{equation*}
    \begin{aligned}
    &\ket{\sigma_n}^{MX^nX'^nR^nW_AB_0} \\
    &\coloneqq \sum_{x^n \in \Sigma^n} \sqrt{p_{x^n}} \ket{\sigma_{x^n}}^{MB_0R^nW_A} \ket{x^n}^{X^n}\ket{x^n}^{X'^n} \\ 
    &\coloneqq \sum_{x^n \in \Sigma^n} \sqrt{p_{x^n}} (U^{A^nJ^nA_0 \to MW_A}_{\mathcal{E}_n}\otimes \1^{R^n})(\ket{\varphi_{x^n}}^{A^nR^n}\ket{j_{x^n}}^{J^n}\ket{\Phi}^{A_0B_0}) \ket{x^n}^{X^n}\ket{x^n}^{X'^n}, 
\\[3ex]      
    &\ket{\xi_n}^{B^nW_AW_BX^nX'^nR^n} \\
    &\coloneqq \sum_{x^n \in \Sigma^n} \sqrt{p_{x^n}} \ket{\xi_{x^n}}^{B^nR^nW_AW_B} \ket{x^n}^{X^n}\ket{x^n}^{X'^n} \\ 
    &\coloneqq \sum_{x^n \in \Sigma^n} \sqrt{p_{x^n}} (U^{MB_0 \to B^nW_B}_{\mathcal{D}_n}\otimes \1^{R^nW_A})\ket{\sigma_{x^n}}^{MB_0R^nW_A} \ket{x^n}^{X^n}\ket{x^n}^{X'^n}. 
    \end{aligned}
\end{equation*}

A common method to certify the quality of data transmission in many
protocols revolves around bounding the trace distance or the fidelity
between the state to be transmitted and the output state of the
protocol (including the references).
In this paper, we consider a general notion of distortion which
include both of these common measures, and we consider an arbitrary
amount of distortion.
Recall that the systems $A, X$ and the source $\rho^{AX}$ are specified
by the compression problem.
The ideal output state of the protocol is $\rho^{BX}$ and
$B \sim A$.
The distortion quantifies how different any state
$\tau^{BX}$ is from the ideal state $\rho^{BX}$,  
formalized as follows:
\begin{definition}~\label{def:distortion}
Let $B, X$ be some fixed quantum systems, and $\rho^{BX}$ a fixed state on $BX$.  
The \emph{distortion} of a quantum state $\tau^{BX}$ with respect to $\rho^{BX}$,  
denoted by $\Delta(\tau^{BX})$, is any non-negative real-valued function $\Delta$ on 
$\tau^{BX}$ which is (1) continuous, (2) convex, and (3) vanishing on $\rho^{BX}$.
For concreteness, we express continuity as 
$|\Delta(\tau^{BX})-\Delta(\tilde{\tau}^{BX})| \leq \|\tau^{BX}-\tilde{\tau}^{BX}\|_1 K_\Delta$ 
where $K_\Delta$ is a constant for a given $\Delta$. 
\end{definition}

We note that a convex function on a finite dimensional space is 
always continuous, but we opt to introduce the notation for the 
Lipschitz constant in the definition above.  We also want to discuss 
the requirement of convexity and continuity separately, in the
hope that convexity can be removed in future studies.  
The third property, $\Delta(\rho^{BX})=0$, ensures that 
all relevant optimizations in this paper are feasible.  
Since the state space is compact, the distortion defined above is always bounded.  
The dependency of the distortion on $\rho^{BX}$ is embedded in the
choice of the function $\Delta$.  
Note that each choice of the distortion determines a specific 
rate-distortion theory.  
A canonical example of a distortion under definition \ref{def:distortion}
is $1-F$ where $F$ is the fidelity of $\tau^{BX}$ with respect to $\rho^{BX}$.  
In this case, the distortion is also \emph{faithful} (that is, 
$\Delta(\tau) > 0$ for $\tau \neq \rho$).  
Another choice of the distortion is the constant function with value
$0$ which is independent of $\rho^{BX}$.  (See further discussion later 
in this section.)  
We derive and present our results for our general definition of distortion.

For the compression setup defined above, with protocol described by 
$\mathcal{E}_n^{A^nJ^nA_0\to M}$ and $\mathcal{D}_n^{MB_0 \to B^n}$, 
we define 
\begin{eqnarray}
    \Delta^{(n)}_{\rm max} (\xi_n^{B^nX^n}) & \coloneqq & \max_{i} \Delta(\xi_n^{B_iX_i}), 
\label{eq:distortion-max} \\
     \Delta^{(n)}_{\rm ave} (\xi_n^{B^nX^n}) & \coloneqq & \frac{1}{n} \sum_{i=1}^n \Delta(\xi_n^{B_iX_i}),  
\label{eq:distortion-ave}
\end{eqnarray}
where $\xi_n^{B_iX_i}$ is the reduced state of $\xi_n^{B^nX^n}$ to the $i$-th system. 
Note that with the above definitions, we quantify distortion letter-wise, 
adopting a worst case or an average case
\textit{local error criterion} in our rate distortion theory. 
The worse case local error criterion 
is suitable for any use of the resulting states that do not
correlate between the copies, and the average case criterion
is suitable in applications such as parameter estimation via averaging.
We use $\Delta^{(n)}$ to denote one of $\Delta^{(n)}_{\rm max}$ or $\Delta^{(n)}_{\rm ave}$.

In the unassisted scenario, 
for a given positive number $D>0$,  
we say that a pair of qubit rate $R$ and distortion $D$ is achievable
if there exists a sequence of codes $\{(\mathcal{E}_n,\mathcal{D}_n)\}_n$ such that 
\begin{align*}
    \Delta^{(n)}(\xi_n^{B^nX^n}) \leq D + \delta_n ~~~{\rm and}~~~ 
    \frac{1}{n} \log \mathrm{dim}(M) \leq R + \eta_n  
\end{align*}
for some vanishing non-negative sequences $\{ \delta_n \}$, $\{ \eta_n \}$.
The \emph{rate-distortion function} is defined as 
\begin{equation*}
    \R(D) \coloneqq \inf\{R : (R,D) \mathrm{\,\,is\,\,achievable}\}. 
\end{equation*}
The entanglement-assisted scenario is essentially the same 
with unlimited $\dim(A_0), \dim(B_0)$ and arbitrary choice of entangled 
state $|\Phi\rangle$ shared on $A_0B_0$, and  
with the notation $\R_{\rm ea}(D)$ for the distortion function.  

We also consider the most general scenario when we count both the 
qubit rate and the entanglement rate.  In this case, we say that a tuple
of qubit rate $R$, entanglement rate $E$, and distortion $D$ is achievable
if there exists a sequence $\{(\mathcal{E}_n,\mathcal{D}_n)\}_n$ 
where $\Delta^{(n)}(\xi_n^{B^nX^n}) \leq D+\delta_n$,  
\begin{align*}
    \frac{1}{n} \log \mathrm{dim}(M) & \leq R+ \eta_n, \\
    \frac{1}{n} \log \mathrm{dim}(A_0) & \leq E+\eta_n,
\end{align*}
for some vanishing non-negative sequences $\{ \delta_n \}$, $\{ \eta_n \}$.
Similarly, the rate-distortion function in this case is defined as 
\begin{equation*}
    \R(D,E) \coloneqq \inf\{R : (R,E,D) \mathrm{\,\,is\,\,achievable}\}. 
\end{equation*}
Note that we use the symbol $E$ to denote some environment system
elsewhere in the paper, and occasionally we use $E$ to denote the
entanglement rate.  Similarly, we use the symbol $R$ to denote some
purifying system; we also use $R$ 
to denote the qubit rate
for compression.  We hope that the different contexts will minimize the
chance of any confusion.

Although the worst-case distortion $\Delta^{(n)}_{\mathrm{max}}$ 
imposes a more stringent error criterion than the average-case 
distortion $\Delta^{(n)}_{\mathrm{ave}}$, remarkably, 
these two error criteria yield the same rate-distortion function 
if Alice and Bob share a sublinear amount of randomness. 
We will discuss this observation in more detail in Section~\ref{sec:discon}.

We illustrate the setup with some examples.  Consider the unassisted
case.  If the distortion is $1-F$ with $F$ being the fidelity, $\R(D)$ 
captures the qubit rate needed to attain the per-copy fidelity $1-D$.   
If the distortion is the constant zero function, $\R(D)=0$ for all $D$.
While mathematically permissible, this
distortion does not lead to an interesting theory.

\section{Main results}~\label{sec:results}
In this section, we present our main results.  We start with  
the entanglement-assisted rate-distortion function in Theorem~\ref{thm:assisted},
followed by the unassisted case in Theorem~\ref{thm:unassisted}.   
The full region of achievable qubit and entanglement rate pairs, as a function of distortion, 
is given by Theorem~\ref{thm:full_rate}. 

\subsection{Entanglement Assisted Rate Distortion Theory}~\label{subsec:assisted}
In this scenario, Alice and Bob have free access to \emph{unlimited} entanglement for the compression task. 
We prove in the following that the entanglement-assisted rate-distortion function for an ensemble source is given by an optimized expression involving the quantum mutual information in single-letter form.

\begin{theorem}[Entanglement-assisted rate-distortion theory]~\label{thm:assisted}
We use the setting in Section \ref{sec:setup} and consider the purification $\ket{\psi}^{AJXX'R}=\sum_x \sqrt{p_x} \ket{\psi_x}^{AR} \ket{j_x}^J \ket{x}^X \ket{x}^{X'}$ of the source $\rho^{AJX}$ (see Eqs.~(\ref{eq:source}) and (\ref{eq:pursource})).  
For both error criteria $\Delta^{(n)}_{\rm max}$ and $\Delta^{(n)}_{\rm ave}$, 
for $D\geq 0$,
the entanglement-assisted rate-distortion function is given by  
    \begin{equation}
        \R_{\rm ea}(D) = \min_{\mathcal{N}^{AJ\to B}} \frac{1}{2}I(B:XX'R)_{\tau},  
    \label{eq:rateexpression1}
    \end{equation}    
    where the minimum is taken over all quantum channels $\mathcal{N}:{AJ\to B}$ such that 
    \begin{equation}
        \Delta \left(  \left(\mathcal{N}\otimes \cid^{X}\right)\left(\rho^{AJX}\right) \right) \leq D,  
    \label{eq:distortpromise1}
    \end{equation}
    and the quantum 
    mutual information is evaluated on the state 
    \begin{align*}
    \tau^{BXX'R}:=(\mathcal{N}\otimes \cid^{XX'R}) \left(\proj{\psi}^{AJXX'R}\right).
    \end{align*}
\end{theorem}

\begin{proof}
Let $f(D)$ denote the solution to the following optimization problem:
$f(D) = \min_{\mathcal{N}^{AJ\to B}} \frac{1}{2}I(B:XX'R)_{\tau}$
where $\tau^{BXX'R}:=(\mathcal{N}\otimes \cid^{XX'R})
\left(\proj{\psi}^{AJXX'R}\right)$ and subject to the constraint
$\Delta \left( \left(\mathcal{N}\otimes
\cid^{X}\right)\left(\rho^{AJX}\right) \right) \leq D$.
We first show achievability ($\R_{\rm ea}(D) \leq f(D)$) and then the converse
($\R_{\rm ea}(D) \geq f(D)$).  

To show achievability, we construct a sequence of codes
$\{(\mathcal{E}_n,\mathcal{D}_n)\}_n$ based on
quantum state redistribution~\cite{Devetak2008a,Yard2009} (Theorem \ref{qsrrecap}) such that 
$\Delta^{(n)}(\xi_n^{B^nX^n}) \leq D + \delta_n$ and 
$\frac{1}{n} \log \mathrm{dim}(M) \leq f(D) + \eta_n$   
for some vanishing non-negative sequences $\{ \delta_n \}$, $\{ \eta_n \}$.
Let $\mathcal{N}^{AJ\to B}$ be a quantum channel satisfying $\Delta (\mathcal{N}^{AJ \to B}(\rho^{AJX})) \leq D$, 
and $U^{AJ \to BE}$ be its Stinespring dilation.  
Suppose that the source generates $n$ copies of the ensemble, resulting in the state 
$\ket{\psi^n}^{A^nJ^nX^nX'^nR^n}$. 
Alice applies $U^{AJ \to BE}$ to each of her $n$ copies of her system, then, the state becomes
\begin{align*}
    \ket{\tau_n}^{B^nE^nX^nX'^nR^n} 
    &\coloneqq \left( 
    (U^{AJ \to BE}\otimes \1^{XX'R})\ket{\psi}^{AJXX'R} \right)^{\otimes n} \\ 
    &\coloneqq \sum_{x^n \in \Sigma^n} \sqrt{p_{x^n}} \ket{\tau_{x^n}}^{B^nR^nE^n}\ket{x^n}^{X^n}\ket{x^n}^{X'^n}. 
\end{align*}
In the above, we have defined  
$\ket{\tau_x} \coloneqq (\1^{R} \otimes U) \ket{\psi_x}^{AR} \ket{j_x}^J$ 
and for $x^n = x_1 x_2 \cdots x_n$, 
$\ket{\tau_{x^n}} \coloneqq \otimes_{i=1}^n  \ket{\tau_{x_i}}$.  
Alice and Bob then perform quantum state redistribution as follows.
Recall that Alice and Bob share entanglement in systems $A_0B_0$.  
At this point, Bob has the system $B_0$, the referee has $R^nX^nX'^n$, and 
Alice has $A_0 B^nE^n$.  Alice tries to transmit $B^n$ to Bob.  
There exists a quantum state redistribution protocol with
qubit rate at most $\frac{1}{2} I(B:RXX') + \eta_n$ and block
error $\epsilon_n$.
Therefore, the state after quantum state redistribution, denoted $\tilde{\tau}_n^{B^nE^nX^nX'^nR^n}$,
satisfies
\begin{equation*}
    \|\tilde{\tau}_n - \tau_n\|_1 \leq \epsilon_n \, 
\end{equation*}
which by monotonicity also implies, for each $i$,  
\begin{equation*}
    \|\tilde{\tau}_n^{B_iX_i} - \tau_n^{B_iX_i}\|_1 \leq \epsilon_n \,. 
\end{equation*}
By continuity of the distortion, we have 
\begin{align*}
    \Delta(\tilde{\tau}_n^{B_iX_i}) \leq D + \epsilon_n \; K_\Delta . 
\end{align*}
The above holds for each $i$, so, both 
$\Delta^{(n)}_{\rm max}$ and $\Delta^{(n)}_{\rm ave}$
are upper-bounded by $D + \epsilon_n \; K_\Delta$.
Since $\frac{1}{2} I(B:RXX')$ is an achievable rate for quantum
state redistribution, $\{\epsilon_n\}$ and $\{\eta_n\}$ are
vanishing.  Letting $\delta_n = \epsilon_n K_\Delta$, $\{\delta_n\}$
is also vanishing.   
Thus $f(D) = \frac{1}{2} I(B:RXX')$ and $D$ are achievable.  

Next, we show the converse.
Suppose a rate $R$ is achievable because of a sequence of codes
$\{(\mathcal{E}_n,\mathcal{D}_n)\}_n$ such that 
$\Delta^{(n)}_{\rm max}(\xi_n^{B^nX^n}) \leq D + \delta_n$ and 
$\frac{1}{n} \log \mathrm{dim}(M) \leq R + \eta_n$   
for some vanishing non-negative sequences $\{ \delta_n \}$, $\{ \eta_n \}$.
We have to show that $R \geq f(D)$.
Consider the following chain of inequalities: 
\begin{align*}
    2n (R+\eta_n)
    &= 2 \log \mathrm{dim}(M)\\ 
    &\geq I(M:R^n X^n X'^n B_0)_{\sigma_n} \\ 
    &= S(M)_{\sigma_n} + S(R^nX^nX'^nB_0)_{\sigma_n} - S(MR^nX^nX'^nB_0)_{\sigma_n}\\ 
    &= S(M)_{\sigma_n} + S(B_0)_{\sigma_n} + S(R^nX^nX'^n)_{\sigma_n} - S(MR^nX^nX'^nB_0)_{\sigma_n}\\
    &\geq S(MB_0)_{\sigma_n} + S(R^nX^nX'^n)_{\sigma_n} - S(MR^nX^nX'^nB_0)_{\sigma_n}\\ 
    & = I(MB_0:R^n X^n X'^n)_{\sigma_n} \\ 
    &\geq I(B^n: R^n X^n X'^n)_{\xi_n} \\ 
    &\geq \sum_{i=1}^n I(B_i:R_iX_iX'_i)_{\xi_n}\\ 
    &\geq 2 n f(D + \delta_n). 
\end{align*}
The second line follows from the dimension upper bound of the quantum mutual information. 
The fourth line holds because preshared entanglement is independent of the data, so, 
system $B_0$ is independent of $R^nX^nX'^n$. 
The fifth line follows from subadditivity of the von Neumann entropy (Lemma~\ref{lem:subadditivity}). 
The seventh line follows from the data-processing inequality (Lemma~\ref{lem:data_processing}). 
Note that from the second to the sixth line, the entanglement assistance is represented 
by $B_0$ which is originally grouped with $R^n X^n X'^n$ and is finally grouped with $M$,  
and $MB_0$ is then turned into $B^n$ in the seventh line.  
The eighth line follows from our generalized lemma for superadditivity
of quantum mutual information (Lemma~\ref{lem:superadditivity}).
To obtain the last inequality, 
let $\mathcal{N}_i^{A_iJ_i \rightarrow B_i}$ denote the quantum channel
on the $i$-th copy of the source obtained by attaching the other $n-1$
copies of the source, followed by applying the protocol $(\mathcal{E}_n,\mathcal{D}_n)$
and then tracing out all but the $i$-th copy.
Since the protocol attains the distortion (Eq.~(\ref{eq:distortion-max})),
$\Delta(\xi_n^{B_iX_i}) \leq D+\delta_n$.
Therefore, the channel 
$\mathcal{N}_i^{A_iJ_i \rightarrow B_i}$
satisfies the constraint 
$\Delta \left( \left( \mathcal{N}_i \otimes
\cid^{X}\right)\left(\rho^{AJX}\right) \right) \leq D+\delta_n$.
Thus $I(B_i:R_iX_iX'_i)_{\xi_n} \geq 2 f(D+\delta_n)$ by the minimizing definition of
the function $f$. 
Finally, as $\delta_n \rightarrow 0$, $f(D+\delta_n) \rightarrow f(D)$ because of the
continuity proved in Appendix \ref{sec:fd}. 
This completes the proof for the converse for the worse case local error criterion.  
Combining with the achieving protocol, the theorem is established.

The theorem also holds for the average local error
criteria.
To see this, the above 
achieving protocol satisfies the global error criteria.  We only need to 
modify the last step of the above converse proof.  We have instead  
$$\sum_{i=1}^n I(B_i:R_iX_iX'_i)_{\xi_n} 
\geq \sum_{i=1}^n 2 f\left(\Delta(\xi_n^{B_iX_i})+\delta_n\right)  
\geq 2 n f\left(\frac{1}{n} \sum_{i=1}^n \Delta(\xi_n^{B_iX_i})+\delta_n\right) 
\geq 2 n f(D+\delta_n)$$ 
where the first inequality holds term-wise, similarly to 
the last step of the original converse proof, but with the $i$-th distortion 
$\Delta(\xi_n^{B_iX_i})+\delta_n$ replacing $D+\delta_n$.  The second inequality is due to 
the convexity of $f(D)$, 
and this is proved in Appendix~\ref{sec:fd}.  The last step
follows from the local error criterion (Eq.~\ref{eq:distortion-ave})
and the fact that, by definition, $f(D)$ is decreasing with $D$.  This
completes the proof for the theorem for the local error criterion.  As
a final remark, with entanglement assistance, the same rate-distortion
function holds for both the worse case and the average case local
error criteria.

\end{proof}

\subsection{Unassisted Rate-Distortion Theory}~\label{subsec:unassisted}
In this subsection, we analyze optimal rate-distortion trade-off when
compressing an ensemble source without any entanglement assistance, 
under the average local error criterion.  
In this case, the unassisted rate-distortion function 
is given by a regularized and optimized expression involving the entanglement of purification. 
\begin{theorem}[Unassisted rate-distortion theory]~\label{thm:unassisted}
We use the setting in Section \ref{sec:setup} and consider the purification $\ket{\psi}^{AJXX'R}=\sum_x \sqrt{p_x} \ket{\psi_x}^{AR} \ket{j_x}^J \ket{x}^X \ket{x}^{X'}$ of the source $\rho^{AJX}$ (see Eqs.~(\ref{eq:source}) and (\ref{eq:pursource})). 
   For $\Delta^{(n)} = \Delta^{(n)}_{\rm ave}$, and for $D>0$, the unassisted rate-distortion function 
   is given by 
    \begin{equation}
        \R(D) =\lim_{k\to\infty} \frac{1}{k} \min_{\mathcal{N}^{A^kJ^k\to B^k}_k} E_p(B^k:X^kX'^kR^k)_{{\tau_k}},  
    \label{eq:proposeduarate}
    \end{equation} 
    where the minimum is taken over all quantum channels $\mathcal{N}_k:{A^kJ^k\to B^k}$ such that 
    \begin{equation}
        \Delta^{(k)}_{\rm ave} \left( (\mathcal{N}_k\otimes \cid^{X^k}) \left( (\rho^{AJX})^{\otimes k}\right) \right) \leq D,
    \label{eq:error-ua}  
    \end{equation}
    and the entanglement of purification is evaluated on the state
    \begin{align}
    \tau_k^{B^kX^kX'^kR^k}:=\left(\mathcal{N}_k\otimes \cid^{X^k{X'}^kR^k}\right) \left( ({\proj{\psi}}^{AJXX'R})^{\otimes k}\right).
    \label{eq:state-ua}
    \end{align}
\end{theorem}

\begin{proof}
We first show achievability, and then the converse.  
We label various expressions on the 
RHS of Eq.~(\ref{eq:proposeduarate}) as follows. 
Let $g_k(D) = \frac{1}{k} \min_{\mathcal{N}^{A^kJ^k\to B^k}_k} E_p(B^k:X^kX'^kR^k)_{{\tau_k}}$ where $\mathcal{N}^{A^kJ^k\to B^k}_k$ is subject to the constraint $\Delta^{(k)}_{\rm ave} \left( (\mathcal{N}_k\otimes \cid^{X^k}) \left( (\rho^{AJX})^{\otimes k}\right) \right) \leq D$, and $g(D) = \lim_{k \rightarrow \infty} g_k(D)$.   
We prove important properties of $g(D)$ and $g_k(D)$ in Appendix \ref{sec:gd}.

To show achievability, we construct a sequence of codes
$\{(\mathcal{E}_n,\mathcal{D}_n)\}_n$ such that 
$\Delta^{(n)}_{\rm ave}(\xi_n^{B^nX^n}) \leq D + \delta_n$ and 
$\frac{1}{n} \log \mathrm{dim}(M) \leq g(D) + \eta_n$   
for some vanishing non-negative sequences $\{ \delta_n \}$, $\{ \eta_n \}$.
In these codes, Alice applies Schumacher compression, and she
processes her state before the compression to minimize the cost.
To construct one of these codes,
first pick $\gamma > 0$.  We use the fact that $g_k(D)$ converges
uniformly to $g(D)$ (see Appendix \ref{sec:gd}) to choose $k$ 
so that $g_k(D) \leq g(D) + \gamma$ for all $D$.
Consider $k$ copies of the source given by the state $\ket{\psi^k}^{A^kJ^kX^kX'^kR^k}$. 
Let $\mathcal{N}_k^{A^kJ^k\to B^k}$ be a quantum channel minimizing
the expression in $g_k(D)$ and satisfying $\Delta^{(k)}_{\rm ave} 
(\mathcal{N}^{A^kJ^k \to B^k}_k((\rho^{\otimes k})^{A^kJ^kX^k})) \leq D$.
Alice applies the quantum channel $\mathcal{N}_k^{A^kJ^k\to B^k}$ to her system $A^kJ^k$. 
Let $U^{A^kJ^k \to B^kE}_{\mathcal{N}_k}$ be the Stinespring dilation of the channel $\mathcal{N}^{A^kJ^k\to B^k}_k$. 
Then, Alice produces a $k$-copy state with purification 
\begin{align*}
    \ket{\tau_k}^{B^kE X^kX'^kR^k} 
    &\coloneqq (U^{A^kJ^k \to B^kE}_{\mathcal{N}_k}\otimes \1^{X^kX'^kR^k}) \ket{\psi^k}^{A^kJ^kX^kX'^kR^k} \\
    &= \sum_{x^k \in \Sigma^k} \sqrt{p_{x^k}} \ket{\tau_{x^k}}^{B^kR^kE}\ket{x^k}^{X^k}\ket{x^k}^{X'^k}. 
\end{align*}
Note that Alice has system $E$ because she locally applies $U^{A^kJ^k \to B^kE}_{\mathcal{N}_k}$. 
She further applies a quantum channel $\Lambda_k^{E \to E_B}$ with isometry $U^{E\to E_BE_A}_{\Lambda_k}$ 
such that $S(B^kE_B)_{\Lambda_k(\tau_k)}$ is minimized. 

Now take a sufficiently large $m$ and consider $n=mk$ copies of the source.  
In the protocol, Alice repeats the above processing $m$ times
to obtain $m$ copies of $\ket{\tau_k}^{B^kE^kX^kX'^kR^k}$, applies $U^{E\to E_BE_A}_{\Lambda_k}$ to each copy, 
and transmits system $(B^k E_B)^{\otimes m}$ to Bob using Schumacher compression~\cite{Schumacher1995,Lo1995}.
Let $\epsilon_m$ be the error in the compression.
Let $\tilde{\tau}_{mk}$ denote the resulting state. 
Then, $\| \tilde{\tau}_{mk} - \tau_{mk} \|_1 \leq \epsilon_m$, so, 
each copy of $\ket{\tau_k}^{B^kE^kX^kX'^kR^k}$ is transmitted with error at most $\epsilon_m$.
As in Theorem~\ref{thm:assisted}, by monotonicity of the trace distance and continuity 
of the distortion, for each $i$, 
\begin{align*}
    \Delta(\tilde{\tau}_k^{B_iX_i}) 
    & \leq D + \epsilon_m \; K_\Delta \,, 
\end{align*}
where $K_\Delta$ is some constant independent of $k$.
Using the known achievable rate for Schumacher 
compression, our rate-distortion code can have rate 
$\frac{1}{k} S(B^kE_B)_{\Lambda_k(\tau_k)} + \gamma_m$, 
so that $\{\epsilon_m\}$ and $\{\gamma_m\}$ are both
vanishing sequences.
We can see that 
\begin{align*}
    S(B^kE_B)_{\Lambda_k(\tau_k)} 
    &= \min_{\Lambda'_k} S(\Lambda'_k(\tau_k))\\ 
    &= E_p(B^k:R^kX^kX'^k)_{\tau_k}, 
\end{align*}
where the minimum is taken over all channels $\Lambda'$ on $E$.
The rate is therefore $g_k(D) + \gamma_m \leq g(D) + \gamma + \gamma_m$.
The quantity $\gamma + \gamma_m$ can be made arbitrarily small as 
$k,m  \rightarrow \infty$.  
Recall that our rate-distortion code has block-length $n=mk$.  Let
  $\delta_n = \epsilon_m \; K_\Delta$.  We have a subsequence of codes
  with distortion $D+\delta_n$ and rate $\frac{1}{k}
  S(B^kE_B)_{\Lambda_k(\tau_k)} + \gamma_m$ where $\{\delta_n\}$ and
  $\{\gamma_m\}$ are both vanishing sequences.
For block-lengths that are not a multiple of $k$, we can use the 
above subsequence of codes for most copies, and use a
noiseless code for the remaining (fewer than $k$) copies.  This 
preserves the distortion 
and the effect on the rate is vanishing. 
Therefore, $g(D)$ is achievable.  

Next, we show the converse, that any achievable rate is at least $g(D)$. 
First pick $\gamma > 0$, and using (2) and (3) proved in 
Appendix \ref{sec:gd}, choose $k$ 
so that $g_k(D) \leq g(D) + \gamma$ for all $D$ and that $g_k(D)$
is continuous.  For a code on $k$-copies with rate $R$ and distortion
$D+\delta$, we have the following chain of inequalities:
\begin{align*}
    kR
    &= \log \mathrm{dim}(M)\\ 
    &\geq S(M)_{\sigma_k}\\
    &\geq E_p(M:R^kX^kX'^k)_{\sigma_k} \\
    &\geq E_p(B^k:R^kX^kX'^k)_{\xi_k}.
\end{align*}
The first inequality follows from the dimension bound of the von Neumann entropy. 
The second inequality follows from Lemma~\ref{lem:upper_bound}. 
The third inequality follows from monotonicity of entanglement of purification (Lemma~\ref{lem:monotonicity}). 
Minimizing over the choice of $\xi_k$ under the distortion constraint for
$D + \delta$, the last expression is no less than $g_k(D+\delta)$.  
Using the continuity of $g_k$ and the uniform convergence of $g_k$ to $g$,
$g_k(D+\delta) - g(D)$ can be made to vanish, as $k$ increases and as
$\delta$ decreases.  

Combining the achievability and the converse, the rate distortion
function for unassisted compression $\R(D)$ is equal to $g(D)$.
Note that the achieving protocol attains worse case local error criterion, but 
the last steps of the converse proof requires the results in Appendix
\ref{sec:gd} which only holds for the average local error criterion. 
\end{proof}

\begin{remark}\label{rem:d0case}
The $D=0$ case requires some special treatments, and the theorem can
still be established in some cases.  First, $g(0)$ remains achievable
for the most general side information system, but the proof for the
converse cannot be established without the continuity and convergence
results in Appendix \ref{sec:gd}.  Second, for blind compression with
trivial side information system, and using a more standard distance
measure for the distortion function $\Delta$ such as the trace
distance or $1-F$ where $F$ is the fidelity, we can adapt other
results in the literature to conclude Theorem \ref{thm:unassisted}.
In particular, references \cite{Koashi2001,Koashi2002} imply $\R(0) =
S(CQ)$ where the state on system $CQ$ is as defined in the
decomposition in Theorem \ref{thm:ki1}.  Meanwhile, when $D=0$, for
each $k$, $g_k$ is minimized over $\mathcal{N}^{A^kJ^k\to B^k}_k$ that
are characterized by Lemma \ref{lem:kiinvar}, and $g_k(0) = S(CQ)$ so
Theorem \ref{thm:unassisted} still holds.
\end{remark}

\begin{remark}
While the rate distortion function for entanglement-assisted
compression is given by a single-letter quantum mutual information
expression, that for unassisted compression is characterized by
multi-letter entanglement of purification.
There are two main hurdles to single-letterize $g(D)$: first, it is
still open if the entanglement of purification is additive, second,
the minimization is over CPTP maps acting on $k$ copies of the source.
\end{remark}

We have derived the rate-distortion function in the most general
unassisted setting, with side information spanning all possible
scenarios between visible and blind compression.
The rate-distortion function is quite complex, and for the rest of
this subsection, we derive simplifications for two special cases,
namely, \textit{visible compression} and \textit{blind compression}.
In the case of visible compression, we show that the rate-distortion
function can be simplified to the entanglement of purification between
Bob's systems $B^k$ and the reference systems $X^k$.  In the case of
blind compression, in the limit of $D\to 0$, we prove that our
multi-letter rate-distortion function becomes single-letter, which
can be related to the previously known optimal rate of blind data
compression.

\subsubsection{Unassisted Visible Compression}

\begin{corollary}~\label{cor:unassisted_visible}
Assuming the set-up in Theorem \ref{thm:unassisted}, but specializing
to visible compression so that $\ket{j_x} = \ket{x}$ in the register
$J$.  The rate distortion function given by Eq.~(\ref{eq:proposeduarate}) 
can be simplified to 
    \begin{equation}~\label{eq:visible_unassisted}
        \R(D) = \lim_{k\to\infty} \frac{1}{k} \min_{\mathcal{N}^{A^kJ^k\to B^k}_k} E_p(B^k:X^k)_{\tau_k}, 
    \end{equation}
with Eqs.~(\ref{eq:error-ua}) and (\ref{eq:state-ua}) continue to define 
$\tau_k$ and the constraints on $\mathcal{N}^{A^kJ^k\to B^k}_k$.
\end{corollary}

\begin{proof}
We will prove that for each $k$, 
\begin{equation}
\min_{\mathcal{N}^{A^kJ^k\to B^k}_k} E_p(B^k:X^kX'^kR^k)_{{\tau_k}} = 
\min_{\mathcal{M}^{A^kJ^k\to B^k}_k} E_p(B^k:X^k)_{\tilde{\tau}_k} \,.
\label{eq:epsimp}
\end{equation}
Note that by the monotonicity of entanglement of purification, the RHS is 
no greater than the LHS. 
It remains to show the opposite inequality.   
Such an inequality means that the correlation between $R^k$ and $B^k$ does
not increase the rate. The physical intuition is that, the error 
constraint Eq.~(\ref{eq:error-ua}) does not involve $R^k$ for mixed-state 
ensembles.  Using her side information $J^k$, Alice's encoding map 
$\mathcal{N}_k$ does not need to operate on the given system $(AR)^k$;  
instead, she can use $J^k$ to generate new systems $(A'R')^k$ 
(with state identical to that in $(AR)^k$) 
before transforming $(JA')^k$ to $B^k$.  The output in $(XB)^k$ still 
satisfies the error constraint, but now $B^k$ is not correlated with
$R^k$.  We provide a rigorous proof in the following.  The value 
of $k$ does not affect the proof idea, so, 
we first focus on $k=1$, with the source 
  \begin{equation}~\label{eq:visible_source}
  \ket{\psi}^{XX'RAJ} \coloneqq \sum_{x \in \Sigma}
  \sqrt{p_x} \ket{x}^X \ket{x}^{X'} \ket{\varphi_x}^{AR} \ket{x}^J .  
  \end{equation}
Consider the RHS of Eq.~(\ref{eq:epsimp}).  
Let $\mathcal{M}^{AJ\to B}$ be the channel minimizing $E_p(B:X)_{\tilde{\tau}}$ while satisfying the error criterion
    \begin{equation}
        \Delta \left( (\mathcal{M} \otimes \cid^{X}) \left( \rho^{AJX} \right) \right) \leq D \,.
    \label{eq:error-ua-k1}
    \end{equation}
    Since $J$ is classical, without loss of generality, the Stinespring dilation for $\mathcal{M}^{AJ\to B}$ 
    can be chosen to be $V^{AJ \to JEB}$ such that $\ket{x}^J$ is preserved.  In other words, 
    $V^{AJ \to JEB} = \sum_x \proj{x}^J \otimes V_x^{A \to EB}$ for some isometries $V_x$.  Let 
    \begin{align}
    \ket{\tilde{\tau}}^{XX'RJEB}:= \left( I^{XX'R} \otimes V^{AJ \to JEB} \right) \ket{\psi}^{XX'RAJ} 
    = \sum_{x \in \Sigma} \sqrt{p_x} \ket{x}^X \ket{x}^{X'} \ket{x}^J \left( (I^R \otimes V_x^{A \to EB}) \ket{\varphi_x}^{AR} \right) .
    \label{eq:state-ua-k1}
    \end{align}
Furthermore, when evaluating $E_p(X{:}B)_{\tilde{\tau}}$, let $\tilde{\Lambda}^{X'JRE \to F}$
be the channel that minimizes $S(BF)_{ (\cid^{BX} \otimes \tilde{\Lambda}^{X'JRE \to F})(\tilde{\tau})} {=}{:} \alpha$.

We now show that the LHS of 
Eq.~(\ref{eq:epsimp}), 
$\min_{\mathcal{N}^{AJ\to B}} E_p(XX'R:B)_{\tau}$, is at most $\alpha$. 
The proof idea is summarized in Fig. \ref{fig:visible-unassisted-pf-idea}. 
%

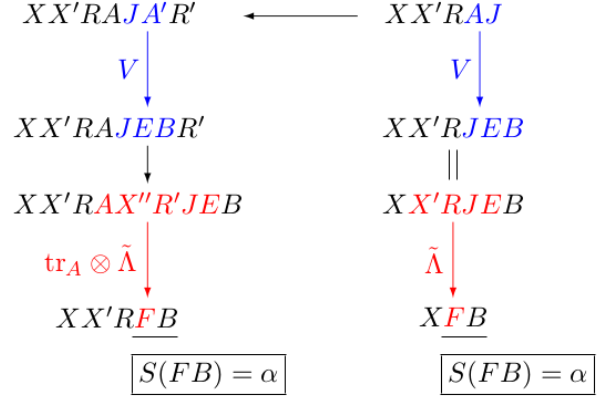
\begin{figure}[ht] 
    \centering
\setlength\unitlength{1.5pt}
\begin{picture}(140,100)(0,0)
\put(10,100){\makebox(50,10){$XX'RA\color{blue}{JA'}\color{black}R'$}}
\color{blue}
\put(45,100){\vector(0,-1){20}}
\put(35,85){\makebox(10,10){$V$}}
\color{black}
\put(10,70){\makebox(50,10){$XX'RA\color{blue}{JEB}\color{black}R'$}}
\put(45,70){\vector(0,-1){10}}
\put(10,50){\makebox(60,10){$XX'R\color{red}AX''R'JE\color{black}B$}}
\color{red}
\put(45,50){\vector(0,-1){20}}
\put(15,35){\makebox(30,10){$\tr_A \otimes \tilde{\Lambda}$}}
\color{black}
\put(22,20){\makebox(30,10){$XX'R\color{red}F\color{black}B$}}
\put(41,20){\line(1,0){12}}
\put(41,05){\framebox(40,10){$S(FB) = \alpha$}}
\put(100,104){\vector(-1,0){30}}
\put(105,100){\makebox(35,10){$XX'R\color{blue}{AJ}$}}
\color{blue}
\put(132,100){\vector(0,-1){20}}
\put(122,85){\makebox(10,10){$V$}}
\color{black}
\put(105,70){\makebox(40,10){$XX'R\color{blue}JEB$}}
\put(124,69){\line(0,-1){8}}
\put(126,69){\line(0,-1){8}}
\put(105,50){\makebox(40,10){$X \color{red} X'RJE \color{black}B$}}
\color{red}
\put(125,50){\vector(0,-1){20}}
\put(115,35){\makebox(10,10){$\tilde{\Lambda}$}}
\color{black}
\color{black}
\put(115,20){\makebox(20,10){$X\color{red}F\color{black}B$}}
\put(122,20){\line(1,0){12}}
\put(122,05){\framebox(40,10){$S(FB) = \alpha$}}
\end{picture}
  \caption{The right column has operations defined by the RHS of Eq.~(\ref{eq:epsimp}).  
  The left column is in terms of the operations in the right column, and presents
  a feasible solution for the minimization of the LHS of Eq.~(\ref{eq:epsimp}) with value 
  same as the RHS.} 
  \label{fig:visible-unassisted-pf-idea} 
\end{figure}

\noindent To show that the LHS of 
Eq.~(\ref{eq:epsimp}), 
$\min_{\mathcal{N}^{AJ\to B}} E_p(XX'R:B)_{\tau}$, is at most $\alpha$, 
we pick a valid (feasible) $\mathcal{N}^{AJ\to B}$ to produce a valid $\tau^{XX'RB}$ 
and a map $\Lambda$ from $\tau^{XX'RB}$'s purifying system to system $F$ so that 
the resulting $S(BF) = \alpha$.  These choices will be in terms of the isometry $V$ 
and the map $\tilde{\Lambda}^{X'JRE \to F}$ developed above in evaluating the RHS of Eq.~(\ref{eq:epsimp}).  
So, let $U^{AJ \to JEB}$ be the Stinespring dilation of $\mathcal{N}^{AJ\to B}$ 
and $U$ acts on the source as follows:  
    \begin{align}
    \ket{\psi}^{XX'RAJ} = & \sum_{x \in \Sigma} \sqrt{p_x} \ket{x}^X \ket{x}^{X'} \ket{\varphi_x}^{AR} \ket{x}^J \label{eq:source-lo}\\ 
    \to & \sum_{x \in \Sigma} \sqrt{p_x} \ket{x}^X \ket{x}^{X'} \ket{\varphi_x}^{AR} \ket{x}^J \ket{\varphi_x}^{A'R'} \label{eq:regen} \\
    \to & \sum_{x \in \Sigma} \sqrt{p_x} \ket{x}^X \ket{x}^{X'} \ket{\varphi_x}^{AR} \ket{x}^J 
                                              \left( (I^{R'} \otimes V_x^{A' \to EB}) \ket{\varphi_x}^{A'R'} \right) {=}{:} ~\ket{\tau}^{XX'ARR'JEB},  
    \label{eq:simbefore}
    \end{align}
where $\ket{x}^J$ generates $\ket{\varphi_x}^{A'R'}$ 
from Eq.~(\ref{eq:source-lo})
to Eq.~(\ref{eq:regen}), 
and then $V^{A'J \to JEB}$ is applied to give Eq.~(\ref{eq:simbefore})
(note that here, $V$ acts on $A'J$).
To see that $\mathcal{N}^{AJ\to B}$ is a valid map, 
note that the final reduced state on $XB$ is the same as in $\ket{\tilde{\tau}}^{XX'RJEB}$ so by 
Eq.~(\ref{eq:error-ua-k1}) $\mathcal{N}^{AJ\to B}$ satisfies the error criterion.  
We can upper-bound the entanglement of purification $E_p(XX'R:B)_\tau$ by 
choosing a channel $\Lambda^{AR'JE \to F}$ to act on $\ket{\tau}^{XX'ARR'JEB}$ 
and evaluating $S(BF)$.  This $\Lambda^{AR'JE \to F}$ creates $X''$ from $J$ 
and then applies $\tilde{\Lambda}^{X''JRE \to F} \otimes \tr^A$.  The 
resulting $S(BF)$ is indeed $\alpha$.  
This completes the proof that the LHS is no bigger than the RHS
in Eq.~(\ref{eq:epsimp}) for $k=1$.
%
The proof for other values of $k$ can be obtained by replacing $A,X,X',R,B,J,X'',R',A'$ by
$A^k,X^k,(X')^k,R^k,B^k,J^k,(X'')^k,(R')^k,(A')^k$ respectively in the
above proof.
\end{proof}

\subsubsection{Unassisted Blind Compression}
Next we consider blind compression, in which Alice does not know the state she receives.  
We use the results in Ref.~\cite{Koashi2002}, and the notations in 
Theorem~\ref{thm:ki1} and Theorem~\ref{thm:unassisted}
throughout this subsection.  
We take $U_{KI}=\1$ without loss of generality.  
Blind compression is 
modeled by a trivial side information system with $\dim(J) = 1$, and the
source simplifies to 
    \begin{equation}~\label{eq:blind_source}
        \ket{\psi}^{AXX'R} = \sum_{x \in \Sigma} \sqrt{p_x} \ket{\varphi_x}^{AR} \ket{x}^X \ket{x}^{X'}.  
    \end{equation}
For simplicity of notation, we omit writing out the identity channels
acting on some states; for example, we write $\mathcal{N}^{A \to
  B}(\proj{\psi})$ or $\mathcal{N}(\proj{\psi})$ omitting writing out
the identity channel on $XX'R$.

Reference \cite{Koashi2002} shows that, under vanishing local or
global error criterion, the optimal rate of blind quantum data
compression is $S(CQ)_{\proj{\psi}}$ where $C$ and $Q$ are the
classical and quantum parts of the ensemble defined in Theorem
\ref{thm:ki1}.  We can connect this result to Theorem \ref{thm:unassisted} 
if we choose the distortion to be 
\begin{equation*}
    \Delta(\tau^{BX}) = 1 - F(\rho^{BX},\tau^{BX}),  
\end{equation*}
where $\rho^{BX}$ is obtained from $\rho^{AX}$ by an identity channel from $A$ to $B$, 
and if we take $D \rightarrow 0$.
Under this distortion, and using Theorem \ref{thm:unassisted} and
Remark \ref{rem:d0case}, $$\R(0) = S(CQ)_{\proj{\psi}} = \lim_{D\to
0} g(D) = \lim_{D\to 0} \lim_{k\to\infty} g_k(D).$$ 
In this subsubsection, using techniques separate from Theorem \ref{thm:unassisted}, 
we show the following, which removes the regularization over $k$ in the above. 
This also gives semicontinuity of $g(D)$ at $D=0$ for the unassisted blind setting
(see Remark \ref{rem:d0case}).
\begin{theorem}~\label{thm:blind_unassisted}
    Consider the blind compression of the source given by Eq.~\eqref{eq:blind_source}. 
    Then, 
    \begin{equation}
        \lim_{D\to 0} \min_{\mathcal{N}^{A\to B}} E_p(B:XX'R)_{\mathcal{N}(\proj{\psi})} = S(CQ)_{\proj{\psi}},  \label{eq:simp-ua-ki}
    \end{equation}
    where the minimum is taken over all channels $\mathcal{N}^{A\to B}$ such that 
    \begin{equation*}
        1 - F(\rho^{BX},\mathcal{N}^{A \to B}(\rho^{AX})) \leq D. 
    \end{equation*}
\end{theorem}

\begin{proof}
Fix $D \geq 0$.  
We first make a useful observation.  
Let $\tilde{\mathcal{N}}^{A\to B}$ be an optimal channel achieving the minimum in the left-hand side
of Eq.~(\ref{eq:simp-ua-ki}).  
Consider the KI operations from Lemma \ref{lem:kiop}
    \begin{align}
        \mathcal{K}^{A\to CQ}_{\mathrm{off}}(\rho^{AX}) &= \omega^{CQX}, \label{eq:koff2}\\
        \mathcal{K}^{CQ\to A}_{\mathrm{on}}(\omega^{CQX}) &= \rho^{AX}. \label{eq:kon2}
    \end{align}
and let $\mathcal{K}^{B\to \hat{B}}_{\mathrm{off}}$ and $\mathcal{K}^{\hat{B}\to B}_{\mathrm{on}}$ 
be the above maps with system $B$ replacing $A$ and $\hat{B}$ replacing $CQ$.  
Let $\mathcal{N}_*^{A\to B} = \mathcal{K}^{\hat{B}\to B}_{\mathrm{on}}\circ \mathcal{K}^{B\to \hat{B}}_{\mathrm{off}} \circ \mathcal{N}^{A\to B}$. 
We now show that $\mathcal{N}_*^{A\to B}$ is also an optimal map. 
First, we check that $\mathcal{N}_*$ also satisfies the distortion condition. 
    \begin{equation*}
        \begin{aligned}
            F(\rho^{BX},\mathcal{N}_*^{A\to B}(\rho^{AX}))
            &= F(\mathcal{K}^{CQ\to B}_{\mathrm{on}}\circ \mathcal{K}^{B\to CQ}_{\mathrm{off}}(\rho^{BX}), \mathcal{K}^{\hat{B}\to B}_{\mathrm{on}}\circ \mathcal{K}^{B\to \hat{B}}_{\mathrm{off}} \circ \tilde{\mathcal{N}}^{A\to B}(\rho^{AX})) \\ 
            &\geq F(\rho^{BX}, \tilde{\mathcal{N}}^{A\to B}(\rho^{AX})) \\ 
            &\geq 1-D, 
        \end{aligned}
    \end{equation*}
    where the first line comes from Eqs.~(\ref{eq:koff2}) and (\ref{eq:kon2}) and the 
    definition of $\mathcal{N}_*$, the second line follows from 
    monotonicity of the fidelity, and the last line holds because   
    the optimal map $\tilde{\mathcal{N}}^{A\to B}$ must satisfy the distortion constraint.  

Second, we show that $\mathcal{N}_*$ also attains the minimum in the left-hand side
of Eq.~(\ref{eq:simp-ua-ki}).  By monotonicity of the entanglement of purification
    \begin{equation*}
        E_p(B:XX'R)_{\mathcal{N}_*(\proj{\psi})} 
        \leq E_p(B:XX'R)_{\tilde{\mathcal{N}}(\proj{\psi})} .
    \end{equation*}
Since $\tilde{\mathcal{N}}^{A\to B}$ is an optimal map achieving the minimum, 
$\mathcal{N}_*^{A\to B}$ is also an optimal map, as claimed. 

    We now show that 
    \begin{equation}
        \lim_{D\to 0} \min_{\stackrel{\mathcal{N}^{A\to B}}{1{-}F(\rho^{BX}\!,\,\mathcal{N}^{A \to B}(\rho^{AX})) \leq D} } E_p(B:XX'R)_{\mathcal{N}(\proj{\psi})} \leq S(CQ)_{\proj{\psi}} 
    \label{eq:ki-connect-first-half}
    \end{equation}
    using the following chain of inequalities
    \begin{align*}
        \min_{\stackrel{\mathcal{N}^{A\to B}}{1{-}F(\rho^{BX}\!,\,\mathcal{N}^{A \to B}(\rho^{AX})) \leq D} } 
        E_p(B:XX'R)_{\mathcal{N}(\proj{\psi})}
        &= E_p(B:XX'R)_{\mathcal{N}_*(\proj{\psi})}  \\[-2ex]
        &\leq E_p(\hat{B}:XX'R)_{\mathcal{K}_{\mathrm{off}}\circ\tilde{\mathcal{N}}(\proj{\psi})} \\[0.5ex]
        &\leq S(\hat{B})_{\mathcal{K}_{\mathrm{off}}\circ\tilde{\mathcal{N}}(\proj{\psi})} \\ 
        &\leq S(CQ)_{\proj{\psi}} - h_2(D) - D\log(\dim(CQ)), 
    \end{align*}    
    where the first line comes from the optimality of $\mathcal{N}_*$, the second
    and third lines come from the monotonicity and dimension bound of the entanglement 
    of purification, and the last line follows from Fannes inequality (Lemma \ref{lem:fannes}).
Thus, in the limit of $D \to 0$, the inequality (\ref{eq:ki-connect-first-half}) holds.

Next, we show the opposite inequality 
    \begin{equation}\label{eq:ki-connect-second-half}
        \lim_{D\to 0} \min_{\stackrel{\mathcal{N}^{A\to B}}{1{-}F(\rho^{BX}\!,\,\mathcal{N}^{A \to B}(\rho^{AX})) \leq D} } 
                    E_p(B:XX'R)_{\mathcal{N}(\proj{\psi})} \geq S(CQ)_{\proj{\psi}}.  
    \end{equation}
Note that the above follows from the results in
\cite{Koashi2002}, since our proof for Theorem \ref{thm:unassisted}
can be applied to the LHS to show that it is an achievable rate under
vanishing local error.  But we also present a direct proof in the 
following. 

Fixed $D \geq 0$ and an optimal map $\tilde{\mathcal{N}}^{A\to B}$ (under constraints as defined before), 
with Stinespring dilation $U_{\tilde{\mathcal{N}}}^{A\to BE}$. 
Let $\ket{\tau_x}^{BER} = (U_{\tilde{\mathcal{N}}}^{A\to BE} \otimes \1^R) \ket{\varphi_x}^{AR}$, 
and 
\begin{equation*}
    \ket{\tau}^{BERXX'} =  \left( U_{\tilde{\mathcal{N}}}^{A\to BE} \otimes \1^{RXX'} \right)  \ket{\psi}^{AXX'R} 
    = \sum_{x \in \Sigma} \sqrt{p_x} \ket{\tau_x}^{BER} \ket{x}^X \ket{x}^{X'}.   
\end{equation*}
By the definition of the entanglement of purification, 
\begin{equation*}
    E_p(B:XX'R)_{\tau} = \min_{\Lambda^{E\to E_B}} S(BE_B)_{\Lambda(\tau)}.  
\end{equation*}
Let $\Lambda_*^{E\to E_B}$ be an optimal channel achieving the minimum above. 
As $D\to 0$, $S(BE_B)_{\Lambda_*(\tau)}$ converges to $S(BE_B)_{\mathcal{M}(\rho)}$, 
where $\mathcal{M}^{A \to BE_B}$ is a quantum channel such that $\tr_{E_B}[\mathcal{M}^{A\to BE_B}(\rho^{AX})] = \rho^{AX}$. 
By Koashi-Imoto's theorem (Theorem~\ref{thm:ki1}), the channel $\mathcal{M}^{A\to BE_B}$ only acts on the redundant part of the state $\rho^{AX}$; 
that is, we can write 
\begin{equation*}
    \mathcal{M}^{A\to BE_B}(\rho^{AX}) 
    \coloneqq \sum_x p_x \sum_{c} p_{c|x} \proj{c}^C \otimes \omega_c^{NE_B} \otimes \rho_{xc}^Q  \otimes \proj{x}^X, 
\end{equation*}
where $\omega_c^{NE_B}$ is a state such that $\tr_{E_B}[\omega_c^{NE_B}] = \omega^N_c$. 
Therefore, 
\begin{align*}
    S(BE_B)_{\mathcal{M}(\rho)}
    &= S(CQNE_B)_{\mathcal{M}(\rho)}\\ 
    &= S(C)_{\mathcal{M}(\rho)} + S(QNE_B|C)_{\mathcal{M}(\rho)} \\
    &= S(C)_{\mathcal{M}(\rho)} + S(Q|C)_{\mathcal{M}(\rho)} + S(NE_B|C)_{\mathcal{M}(\rho)} \\
    &\geq S(C)_{\mathcal{M}(\rho)} + S(Q|C)_{\mathcal{M}(\rho)} \\ 
    &= S(CQ)_{\mathcal{M}(\rho)} \\ 
    &= S(CQ)_{\rho} \\ &= S(CQ)_{\proj{\psi}}  
\end{align*}
where the third line is obtained because the state on $QNE_B$ is a product state 
when we conditioned on system $C$; in particular, for each $c$, the state on $QNE_B$ 
is $\omega_c^{NE_B} \otimes \sum_x p_x \rho_{cx}^{Q}$. 
Therefore, the inequality (\ref{eq:ki-connect-second-half}) holds. 
\end{proof}

\subsection{Full Rate Region of Rate Distortion Theory}~\label{subsec:rate_region}
An ultimate goal in quantum information theory is to understand the fundamental performance and cost of a task.  
Here, we generalize our arguments from the previous sections to provide the full qubit-entanglement rate region of rate distortion coding for ensemble sources. 

\begin{theorem}[Full rate region of rate-distortion compression]~\label{thm:full_rate}
    Rate-distortion coding of a given source $\ket{\psi}^{AJXX'R}$ with distortion $D$ is achievable if and only if 
    its qubit rate $R$ and entanglement rate $E$ satisfy
    \begin{align}
        R &\geq \frac{1}{k} \; \frac{1}{2} I(B^kE_B:R^kX^kX'^k)_{\Lambda_k(\tau_k)}, \label{eq:qrate} \\
        R+E &\geq \frac{1}{k} \; S(B^kE_B)_{\Lambda_k(\tau_k)}, \label{eq:erate}
    \end{align}
for some $k$ and for the state $\tau_k$ defined in the statement of Theorem \ref{thm:unassisted}, 
involving a channel $\mathcal{N}_k^{A^kJ^k\to B^k}$ satisfying the distortion 
condition in the statement of Theorem \ref{thm:unassisted}, and 
some channel $\Lambda_k^{E \to E_B}$.
\end{theorem}

\begin{proof}
We first show achievability, and then the converse.

The achievability of the given rate-region can be shown by considering a 
protocol very similar to the one used in the proof of Theorem \ref{thm:unassisted}, 
so we only describe the two differences: 
(1) Alice's processing need not minimize $S(B^kE_B)_{\Lambda_k(\tau_k)}$ here. 
(2) Instead of Schumacher compression, $m$ copies of the $k$-copy state are transmitted using 
quantum state redistribution~\cite{Devetak2008a,Yard2009} with error $\epsilon_m$.  
This quantum state redistribution can be performed with 
qubit rate $R+\eta_m$ and entanglement rate $E+\eta_m$ if 
\begin{align*}
    R &\geq \frac{1}{k} \; \frac{1}{2} I(B^kE_B:R^kX^kX'^k)_{\Lambda_k(\tau_k)}, \\
    R+E &\geq \frac{1}{k} \; S(B^kE_B)_{\Lambda_k(\tau_k)} \,    
\end{align*}
for some vanishing non-negative sequence $\{\eta_m\}$.  
The same continuity argument in Theorem \ref{thm:unassisted} also 
applies to show that the protocol satisfies the distortion condition 
up to a vanishing amount as $\epsilon_m \rightarrow 0$.  
Finally, for large block-lengths that are not multiples of $k$, 
we can again use the above subsequence of codes on most of the 
copies, and transmit the remainder noiselessly, with vanishing 
increase to the above rates, similarly to the discussion in the 
proof of Theorem \ref{thm:unassisted}.

For the converse, consider a fixed $k$ 
and a compression protocol $(\mathcal{E}_k, \mathcal{D}_k)$
with qubit rate $R$ and entanglement rate $E$, and satisfying the
distortion condition given by Definition \ref{def:distortion} and
Eq.~(\ref{eq:distortion-ave}), with $D$ replaced by $D+\delta_k$  
for a small $\delta_k$ (and a large $k$).  
For this protocol, $\sigma_k$, 
$\xi_k$ are as defined in Section \ref{sec:setup} (with $k$ instead of $n$
for the block size).  
We need to propose quantum channels $\mathcal{N}_k^{A^kJ^k\to B^k}$
with Stinespring dilation $U_{\mathcal{N}_k}^{A^kJ^k \to B^kE}$ and
$\Lambda_k^{E \to E_B}$ with Stinespring dilation 
$U_{\Lambda_k}^{E \to E_AE_B}$ so that $\mathcal{N}_k^{A^kJ^k\to B^k}$ 
satisfies the distortion condition (with $D$ replaced by $D+\delta_k$), and the inequalities 
(\ref{eq:qrate})
and 
(\ref{eq:erate})
on the rates $R,E$ are satisfied.  
To this end, we take $\mathcal{N}_k^{A^kJ^k\to B^k}$ to be the composition of 
the following steps:  \\
(1) appending the entangled state $\Phi^{A_0B_0}$ to $(\rho^{\otimes k})^{A^kJ^kX^k}$, \\ 
(2) applying the encoding operation $\mathcal{E}_k^{A^kJ^kA_0 \to M}$ (with environment system $W_A$), \\
(3) applying the decoding operation $\mathcal{D}_k^{MB_0 \to B^k}$ (with environment system $W_B$), \\ 
and $\mathcal{N}_k^{A^kJ^k\to B^k}$ has environment system $E=W_AW_B$.  We take 
$\Lambda_k^{E \to E_B}$ as the quantum channel that first discards system $W_A$ and then 
renames $W_B$ as $E_B$. 

We now have the following chain of inequalities for $k$ copies with qubit rate $R$: 
\begin{align*}
    2kR
    &= 2 \log \mathrm{dim}(M)\\ 
    &\geq I(MB_0:R^k X^k X'^k)_{\sigma_k} \\ 
    &= I(B^kW_B: R^k X^k X'^k)_{\xi_k}, 
\end{align*}
where the first inequality is obtained with the same arguments 
for the first 6 lines of the chain of inequalities in the converse 
proof of Theorem~\ref{thm:assisted}. 
We also have the following chain of inequalities with qubit rate $R$ and entanglement rate $E$: 
\begin{align*}
    kR + kE
    &= \log \mathrm{dim}(M) + \log \mathrm{dim}(B_0)\\ 
    &\geq S(M)_{\sigma_k} + S(B_0)_{\sigma_k} \\
    &\geq S(MB_0)_{\sigma_k} \\  
    &= S(B^kW_B)_{\xi_k}. 
\end{align*}
Identifying $\xi_k$ as $\Lambda_k(\tau_k)$ and $W_B$ as $E_B$, we obtain
\begin{align*}
    2kR & \geq I(B^kE_B: R^k X^k X'^k)_{\Lambda_k(\tau_k)}, \\ 
    kR + kE& \geq S(B^kE_B)_{\Lambda_k(\tau_k)}. 
\end{align*}
To see that the above bounds do not change abruptly when $D+\delta_k$ 
is replaced by $D$, we use an argument similar to that used in Appendix 
\ref{sec:gd}, taking $l$ copies of a $k$-copy protocol for $D_1$ and $1$ 
copy of a $k$-copy protocol for $D_2$, and noting that the overall 
achievable distortion parameter is convex, while the RHS of the above
bounds are additive.  This removes the abrupt changes in the bounds 
as $l$ becomes large.  

\end{proof}

\begin{remark}
Comparing the rates in Theorem \ref{thm:full_rate}
with those in Theorems 
\ref{thm:assisted} and \ref{thm:unassisted}.  
The qubit rate $\frac{1}{k} I(B^kE_B:R^kX^kX'^k)_{\Lambda_k(\tau_k)}$
in Theorem \ref{thm:full_rate}
is lower-bounded by the assisted rate
$\frac{1}{2}I(B:XX'R)_{\mathcal{N}(\proj{\psi})}$ from Theorem
\ref{thm:assisted} (by monotonicity when dropping $E_B$ and by
superadditivity).
Choosing a trivial $E_B$ can increase $R+E$ but that is not a concern
with entanglement asssistance.  Operationally, if Alice and Bob share
enough entanglement, Alice need not apply her local operation
$\Lambda^{E \to E_B}$ to achieve the optimal qubit rate.
Meanwhile, minimizing the rate sum $R+E$ in 
Theorem \ref{thm:full_rate}
matches the unassisted rate
in Theorem \ref{thm:unassisted}; operationally, it says that
transmitting the required entanglement and performing quantum 
state redistribution is as optimal as Schumacher compression 
used in Theorem \ref{thm:unassisted}.
\end{remark}

\section{Discussion and Conclusion}
\label{sec:discon}

First, we summarize our main results.
We investigated quantum rate-distortion compression with and without entanglement. 
We assumed a quantum mixed-state ensemble source with side information so that our analysis covers both visible and blind compression.
First, we derived the rate-distortion function for entanglement assisted compression, and we proved that the optimal rate can be expressed using a single-letter formula of quantum mutual information. 
Second, we showed that the rate-distortion function for the unassisted case is expressed in terms of a regularized entanglement of purification. We showed interesting implications of this result to visible and blind scenarios. 
In the visible scenario, the distortion function can be simplified, and it only depends on Bob's system and the reference system. 
In the blind case, we proved that our multi-letter rate distortion function can be single-letterized in the limit of $D\to 0$. 
Finally, we found the full qubit-entanglement rate region for rate-distortion theory of mixed states. 
Thus, we characterized rate-distortion compression of mixed states, which has not been covered before. 
We believe our work furthers the understanding of the limits of quantum data compression with finite approximations. 

Second, we explore several observations and implications of our results. 
\begin{itemize}
\item We can compare our results to those in references
    \cite{Datta2013b,Datta2013c}
    for a pure-state source 
    \begin{equation}
        \rho^{AX} = \sum_{x\in\Sigma} p_x \proj{\varphi_x}^A\otimes \proj{x}^X
    \end{equation}
    with purification 
    \begin{equation}
        \ket{\psi}^{AXX'} = \sum_{x\in\Sigma} \sqrt{p_x} \ket{\varphi_x}^A\ket{x}^X\ket{x}^{X'}. 
    \end{equation} 
    We omit the side-information system $J$ in this comparison since \cite{Datta2013b,Datta2013c} 
    considered side-information differently from our treatment.  Note also that the system $R$ 
    is trivial for a pure state ensemble.   
    
    First, we note that the figure of merit in this paper is different from 
    that in Refs.~\cite{Datta2013b,Datta2013c},  
    even though both adopted the (average-case) local error criterion.
    The distortion function $\Delta$ in this paper evaluates the quality of the data transmission on the joint system $AX$ (Definition~\ref{def:distortion}).  
    On the other hand, the distortion function 
    in Refs.~\cite{Datta2013b,Datta2013c} 
    is with respect to the system $AXX'$. (Note that we have adjusted the notation accordingly.)
    Our distortion condition is less stringent in that we do not require the coherence over the label $x$ to be preserved. 

    Now, we compare our rate-distortion functions to those in Refs.~\cite{Datta2013b,Datta2013c}, and  
    observe some similarities.  
    For example, in the entanglement-assisted scenario, 
    our rate-distortion function for the pure-state source is 
    \begin{equation}
        \label{eq:ea_distortion_pure}
        \mathscr{R}_{\mathrm{ea}}(D) = \min_{\mathcal{N}^{A\to B}} \frac{1}{2}I(B:XX')_{(\mathcal{N}\otimes \mathcal{I})(\proj{\psi})}, 
    \end{equation}    
    where the minimum is taken over all quantum channels $\mathcal{N}:{A\to B}$ such that 
    \begin{equation}
        \label{eq:ea_condition_pure}
        \Delta \left(  \left(\mathcal{N}\otimes \mathcal{I}^{X}\right)\left(\rho^{AX}\right) \right) \leq D. 
    \end{equation}
    Meanwhile, the rate-distortion function in Refs.~\cite{Datta2013b,Datta2013c} 
    has the same form as Eq.~\eqref{eq:ea_distortion_pure}, 
    but the minimization is taken over 
    quantum channels $\mathcal{N}:{A\to B}$ such that 
    \begin{equation}
        \label{eq:ea_condition_pure_prev}
        \Delta \left(  \left(\mathcal{N}\otimes \mathcal{I}^{XX'}\right)\left(\proj{\psi}^{AXX'}\right) \right) \leq D. 
    \end{equation} 
    Even though the two rate-distortion functions have almost identical expressions, 
    they can have potentially different values due to the different 
    distortion conditions (Eqs.~\eqref{eq:ea_condition_pure} and \eqref{eq:ea_condition_pure_prev}). 
    Very similar relations between the two results also hold for 
    the unassisted case via a similar argument.

    \item The achievability proofs of Theorems~\ref{thm:assisted}, \ref{thm:unassisted}, and \ref{thm:full_rate} 
    indicate that the optimal rates for all our rate-distortion theories 
    can always be achieved by faithfully transmitting an appropriately distorted source. 
    Hence, in the rate-distortion coding, it is sufficient for the sender to prepare a good approximation of the given ensemble.

    \item 
    If Alice and Bob are allowed to use shared randomness for their protocol, 
    the worst-case distortion $\Delta^{(n)}_{\mathrm{max}}$ and 
    the average-case distortion $\Delta^{(n)}_{\mathrm{ave}}$ 
    result in the same rate distortion function. 
    To see this, first note that since $\Delta^{(n)}_{\mathrm{ave}} \leq \Delta^{(n)}_{\mathrm{max}}$ by definition, 
    the rate distortion function for $\Delta^{(n)}_{\mathrm{ave}}$ 
    is no greater than that for $\Delta^{(n)}_{\mathrm{max}}$. 
    Next, to show the opposite inequality, 
    given any protocol $(\mathcal{E}_n, \mathcal{D}_n)$, 
    we construct a new protocol $(\tilde{\mathcal{E}}_n, \tilde{\mathcal{D}}_n)$ 
    in which Alice applies a random cyclic permutation $\pi$ on $A_1 J_1$, $\cdots$, 
    $A_n J_n$ \emph{before} $\mathcal{E}_n$, and Bob reverts $\pi$ \emph{after} $\mathcal{D}_n$.  (More precisely, $\pi = (1,\cdots,n)^k$ for a uniformly chosen $k \in \{0,\cdots,n-1\}$.)   
    Note that 
    (i) the qubit and the entanglement rates are the same as in the original protocol, and   
    (ii) the worst-case distortion for the new protocol is no more than the average-case distortion of the original protocol. 
    The latter holds because the new protocol results in the state $\tilde{\xi}^{B^nX^n}$ whose marginals  
    $\tilde{\xi}_n^{B_iX_i} = \frac{1}{n} \sum_{j=1}^n \xi_n^{B_jX_j}$ for all $i$, 
    where $\xi_n^{B^nX^n}$ is the resulting state by the original protocol,   
    because in the new protocol, 
    Alice and Bob just randomly reindex the subsystems 
    while keeping the original encoding and decoding. 
    Thus, 
    \begin{equation*}
        \Delta^{(n)}_{\mathrm{max}}(\tilde{\xi}_n^{B^nX^n})
        = \max_{i} \Delta(\tilde{\xi}_n^{B_iX_i}) 
        = \Delta\left(\frac{1}{n} \sum_{j=1}^n \xi_n^{B_jX_j}\right)
        \leq \frac{1}{n} \sum_{j=1}^n \Delta(\xi_n^{B_jX_j}) 
        = \Delta^{(n)}_{\mathrm{ave}}(\xi_n^{B^nX^n}), 
    \end{equation*}
    where the inequality follows from the convexity of the distortion $\Delta$ (Definition~\ref{def:distortion}). 
    Note that in the new protocol, Alice and Bob only need $\log_2 n$ bits of shared randomness, 
    so only negligible amount of shared randomness is required asymptotically. 
    In particular, when Alice and Bob use a part of their shared entanglement as the shared randomness,  
    they can still achieve the same asymptotic rate region for the average-case distortion and  
    the worst-case distortion without sacrificing the entanglement rate 
    since the entanglement rate consumed for the randomness is $(\log_2n)/n$.

    \item 
    Our results raise the natural question of the computability (analytically or numerically) of our rate-distortion functions for simple examples. 
    To see this, let us formalize the rate-distortion functions as optimization problems.  
    For example, the entanglement-assisted rate-distortion function $\mathscr{R}_{\mathrm{ea}}(D)$ can be obtained as the optimal value of the following problem:  
    \begin{equation}
        \begin{array}{cl}
            \mathrm{minimize} & \frac{1}{2}I(B:XX'R)_{\tau} \\
            \mathrm{subject\,\,to} & \tau^{BXX'R}=(\mathcal{N}\otimes \mathcal{I}^{XX'R}) \left(\proj{\psi}^{AJXX'R}\right) \\
            & \Delta \left(\left(\mathcal{N}\otimes \mathcal{I}^{X}\right)\left(\rho^{AJX}\right) \right) \leq D \\  
            & \mathcal{N}:\mathrm{A\,\,quantum\,\,channel\,\,from\,\,}AJ \mathrm{\,\,to\,\, } B.
        \end{array} 
    \end{equation}    
    This problem is convex because
    \begin{enumerate}
        \item $I(B:RXX')_{\tau}$ is convex with respect to $\cN$ (Lemma~\ref{lem:convexity}), 
        \item the distortion function $\Delta$ is convex with respect to $\cN$ (Definition~\ref{def:distortion}), and
        \item the set of quantum channels from $AJ$ to $B$ is convex.  
    \end{enumerate}
    Hence, we can use convex-optimization algorithms to evaluate this rate-distortion function numerically.
    Since $A$, $J$, and $B$ are finite-dimensional, the domain of optimization is compact and 
    the optimal value can be attained (note that a feasible solution is the identity channel from $A$ to $B$ and tracing out $J$).  
    
    On the other hand, however, in the unassisted case, it is hard to compute the rate-distortion function in general because the expression involves regularization. 
    Moreover, even if we consider the expression with some finite $k$ instead of the limit of $k\to \infty$, 
    the corresponding optimization problem is not convex since the entanglement of purification is neither convex nor concave~\cite{Terhal2002}.

\item 
So far, we have focused on the simple pure state side information 
in our framework and derivations.  
The most general side information for an ensemble source can be
written as $\{p_x, \rho_x^{AJ}\}_{x \in \Sigma}$, where $\rho^{AJ}_x$
is an extension of $\rho_x^{A}$ ($\tr_J \rho^{AJ}_x = \rho_x^{A}$).
In other words, the only restriction for side information is that,
removing it returns the original state.
Note that $\rho^{AJ}_x$ can be entangled across $A$ and $J$, and 
the above definition includes the pure state side information 
as a special case.  
Theorems~\ref{thm:assisted}, \ref{thm:unassisted}, and
\ref{thm:full_rate}, and results in Appendices \ref{sec:fd}
and \ref{sec:gd} can be extended to the most general 
side information, if $R$ is revised to purify $AJ$ for each $x$.  
Our proofs for pure state side information then apply to the most
general side information without change, because all the achieving
protocols and distortion maps in the rate distortion functions operate
on copies of $AJ$ without assumptions on the state on $AJ$.  Detailed
verifications are left to the interested readers.

\item Our extension of quantum rate-distortion theory to mixed states 
ensembles also provides some conclusions that cannot be obtained from
Refs.~\cite{Datta2013b,Datta2013c}.  One such result is proving the fact that 
the rate-distortion function is unaffected by the addition or removal 
of redundant parts to an ensemble.  
Here, we omit writing out the identity channel $\mathcal{I}$ 
when considering a quantum channel acting on a subsystem of a 
given state. More specifically, let the distortion $\Delta$ be 
$1-F$, and suppose 
we have a source
\begin{equation}~\label{eq:redundant_source}
    \omega^{C N Q J X} =
    \sum_x p_x \sum_c p_{c|x} \; \proj{c}^{C} \otimes \omega_c^{N} \otimes \rho_{cx}^{Q} 
    \otimes \proj{j_x}^{J} \otimes \proj{x}^X 
\end{equation}
with purification 
\begin{equation}~\label{eq:redundant_source_purified}
    \ket{\omega}^{C N Q J X X' R_C R_N R_Q}
    \coloneqq \sum_x \sqrt{p_x} \sum_c \sqrt{p_{c|x}} \ket{c}^{C}\ket{c}^{R_C} \ket{\omega_c}^{NR_N} \ket{\rho_{cx}}^{QR_Q} \ket{j_x}^{J} \ket{x}^X\ket{x}^{X'}, 
\end{equation}
where $R_C$ purifies $C$, $R_N$ purifies $N$, $R_Q$ purifies $Q$, 
and $X'$ purifies the distribution over $X$. 
We prove that the entanglement-assisted rate-distortion function of $\omega^{CNQX}$ is unchanged when we remove the redundant part $N$. 
\begin{theorem}~\label{thm:KI_redundant}
    Consider the distortion $\Delta = 1-F$ and the source $\omega^{C N Q J X}$ 
    with purification $\ket{\omega}^{C N Q J X X' R_C R_N R_Q}$ as defined above. 
    Recall that the rate-distortion function for the entanglement-assisted rate-distortion function $\mathscr{R}_{\mathrm{ea}}(D)$ is given by 
    \begin{equation}~\label{eq:KI_ea_rate}
        \mathscr{R}_{\mathrm{ea}}(D) = \min_{\mathcal{N}^{CNQJ \to CNQ}} \frac{1}{2} I(CNQ: R_CR_NR_QXX')_{\mathcal{N}(\proj{\omega})}, 
    \end{equation}
    where the minimum is taken over all quantum channels $\mathcal{N}: CNQJ \to CNQ$ 
    with $1 - F \left(\omega^{CNQX},  \mathcal{N}\left(\omega^{CNQJX}\right) \right) \leq D$. 
    The rate distortion function for this source can be simplified to 
    \begin{equation}
        \mathscr{R}_{\mathrm{ea}}(D) = \min_{\mathcal{M}^{CQJ \to CQ}} \frac{1}{2} I(CQ: R_CR_QXX')_{\mathcal{M}(\proj{\tilde{\omega}})}, 
    \end{equation}
    where 
    \begin{equation}
        \ket{\tilde{\omega}}^{CQJXX'R_CR_Q} \coloneqq \sum_x \sqrt{p_x} \sum_c \sqrt{p_{c|x}} \ket{c}^{C}\ket{c}^{R_C} \ket{\rho_{cx}}^{QR_Q} \ket{j_x}^{J} \ket{x}^X\ket{x}^{X'}
    \end{equation}
    and the minimum is taken over all quantum channels $\mathcal{M}: CQJ \to CQ$ 
    with $1 - F \left(\omega^{CQX},  \mathcal{M}\left(\omega^{CQJX}\right) \right) \leq D$.
\end{theorem}
We provide a detailed proof in Appendix \ref{app:redundant-part-pf}.

A similar conclusion can be made in the unassisted case, with a detailed proof given in Appendix \ref{app:redundant_un}.  
    \begin{theorem}~\label{thm:KI_redundant_un}
        Consider the source $\omega^{C N Q J X}$ 
        with purification $\ket{\omega}^{C N Q J X X' R_C R_N R_Q}$ as defined above. 
        Recall that the unassisted rate-distortion function $\mathscr{R}(D)$ is given by 
        \begin{equation}
            \mathscr{R}(D) = \lim_{k\to\infty} \frac{1}{k} \min_{\mathcal{N}_k} E_p(C^kN^kQ^k: R_C^kR_N^kR_Q^kX^kX'^k)_{\mathcal{N}_{k}\left(\proj{\omega}^{\otimes k}\right)}, 
        \end{equation}
        where the minimum is taken over all quantum channels $\mathcal{N}_k: C^kN^kQ^kJ^k \to C^kN^kQ^k$ 
        with 
        \begin{equation}
            1 - \frac{1}{k}\sum_{j=1}^{k} F \left(\omega^{C_jN_jQ_jX_j}, \left(\mathcal{N}_{k}\left(\proj{\omega}^{\otimes k}\right)\right)^{C_jN_jQ_jX_j} \right) \leq D. 
        \end{equation}
        Then, the rate distortion function for this source can be simplified to 
        \begin{equation}
            \mathscr{R}(D) = \lim_{k\to\infty} \frac{1}{k} \min_{\mathcal{M}_k} E_p(C^kQ^k: R_C^kR_Q^kX^kX'^k)_{\mathcal{M}_{k}\left(\proj{\tilde{\omega}}^{\otimes k}\right)}, 
        \end{equation}
        where the minimum is taken over all quantum channels $\mathcal{M}_k: C^kQ^kJ^k \to C^kQ^k$ 
        with 
        \begin{equation}
            1 - \frac{1}{k}\sum_{j=1}^{k} F \left(\omega^{C_jQ_jX_j}, \left(\mathcal{M}_{k}\left(\proj{\tilde{\omega}}^{\otimes k}\right)\right)^{C_jQ_jX_j} \right) \leq D. 
        \end{equation}
    \end{theorem}

    The above discussion relies on the structure of the ensemble, so, 
    it is more difficult to consider the most general side information. 
    It turns out that the above results concerning 
    the addition and removal of redundant parts  
    in the KI decomposition indeed still apply to general side information 
    that entangles only with the quantum part $Q$. 
    More specifically, suppose that we have a source 
    \begin{equation*}
        \omega^{C N Q J X} =
        \sum_x p_x \sum_c p_{c|x} \; \proj{c}^{C} \otimes \omega_c^{N} \otimes \rho_{cx}^{QJ} 
        \otimes \proj{x}^X 
    \end{equation*}
    with purification 
    \begin{equation*}
        \ket{\omega}^{C N Q J X X' R_C R_N R_Q}
        \coloneqq \sum_x \sqrt{p_x} \sum_c \sqrt{p_{c|x}} \ket{c}^{C}\ket{c}^{R_C} \ket{\omega_c}^{NR_N} \ket{\rho_{cx}}^{QJR_Q} \ket{x}^X\ket{x}^{X'}. 
    \end{equation*}
    Then, the statements in Theorems~\ref{thm:KI_redundant} and \ref{thm:KI_redundant_un} hold for this source as well. 
    For example, the entanglement-assisted rate distortion function 
    $\mathscr{R}_{\mathrm{ea}}(D)$ with the fidelity distortion $\Delta = 1 - F$
    is given as 
    \begin{equation*}
        \mathscr{R}_{\mathrm{ea}}(D) = \min_{\mathcal{M}^{CQJ \to CQ}} \frac{1}{2} I(CQ: R_CR_QXX')_{\mathcal{M}(\proj{\tilde{\omega}})} 
    \end{equation*}
    with 
    \begin{equation*}
        \ket{\tilde{\omega}}^{CQJXX'R_CR_Q} \coloneqq \sum_x \sqrt{p_x} \sum_c \sqrt{p_{c|x}} \ket{c}^{C}\ket{c}^{R_C} \ket{\rho_{cx}}^{QJR_Q} \ket{j_x}^{J} \ket{x}^X\ket{x}^{X'}, 
    \end{equation*}
    where the minimum is taken over all quantum channels $\mathcal{M}: CQJ \to CQ$ 
    with $1 - F \left(\omega^{CQX},  \mathcal{M}\left(\omega^{CQJX}\right) \right) \leq D$.
    The proof is almost the same as that of Theorem~\ref{thm:KI_redundant}, 
    for the channel $\mathcal{M}^{CQJ \to CQ}$ in Theorem~\ref{thm:KI_redundant} 
    can take any state on $CQJ$, in particular, even if the state on $QJ$ is entangled. 
    The unassisted rate distortion function can also be characterized similarly. 
    
\end{itemize}

We end the paper with several interesting future directions.
In Theorem~\ref{thm:blind_unassisted}, we proved a relation between Koashi and Imoto's optimal blind compression rate 
and our rate-distortion function. 
On the other hand, since the optimal visible compression rate was analyzed under global error criterion, there is no direct connection between the optimal rate and our rate-distortion function in the visible case. 
In addition, it would be interesting to see if our rate-distortion functions are additive or not. Some of these rate-distortion functions are expressed using multi-letter formulae, and it would be highly useful if there are simplified characterizations.
In this paper, we consider the average (and worst-case) local error, for which we obtained continuous rate distortion functions in all the settings considered, so, strong converse does not hold.  
For mixed state compression under global error, significant progress in obtaining a strong converse has been reported in the unassisted case \cite{Khanian2022}, and results in \cite{Bennett2014a} provide a promising direction for the assisted case.  Further research in these, and also in the general $(Q,E)$ rate region are important open problems to be tackled.
More recently, an alternative approach to
rate-distortion theory has been proposed in \cite{Atif2023}. Their
framework differs significantly from most of the existing work:
a distortion map is applied to the reference system, and it will be
interesting to find potential connections between their framework and ours.
Finally, we believe that our approach taken in this paper can be applied to other setups that involves data processing of mixed states. 
It would be highly important to see what type of trade-off can be seen in other quantum information processing setups. 


\section*{Acknowledgements}
We are grateful to Masato Koashi for suggesting the equivalence of 
the average-case and the worst-case distortions using random 
permutations. 
We thank Eric Chitambar for bringing up the consideration of 
entangled side information. 
We thank Andreas Winter and Robert K\"{o}nig for additional helpful discussions. 
Z.~B.~K. was supported by
the DFG cluster of excellence 2111 (Munich Center for Quantum
Science and Technology) and the Marie Sk{\l}odowska-Curie Actions 
(MSCA) European Postdoctoral Fellowships (Project 101068785-QUARC). 
K.\ K. was supported by a Mike and Ophelia Lazaridis Fellowship and a 
Funai Overseas Scholarship,
and K.\ K and D.\ L are supported by 
an NSERC discovery grant.

\printbibliography[heading=bibintoc]

\appendix
\subsection{Convexity and continuity of the entanglement-assisted rate-distortion function}~\label{sec:fd}
\begin{lemma}~\label{lem:rate_convex}
   Let $f(D)$ be defined as in the beginning of the proof of  
   Theorem \ref{thm:assisted}.  
   For $D_1,D_2 > 0$ and $\lambda \in [0,1]$, 
   \begin{equation*}
        f(\lambda D_1 + (1-\lambda)D_2) 
        \leq \lambda f(D_1) + (1-\lambda) f(D_2).  
   \end{equation*}
\end{lemma} 
\begin{proof}
    For any state $\psi^{AJXX'R}$, let $\mathcal{N}^{AJ\to B}_1$ and $\mathcal{N}^{AJ\to B}_2$ 
    be quantum channels achieving 
    \begin{align*}
        f(D_1) &= \frac{1}{2} I(B:XX'R)_{(\mathcal{N}_1 \otimes \cid)(\psi)}, \\ 
        f(D_2) &= \frac{1}{2} I(B:XX'R)_{(\mathcal{N}_2 \otimes \cid)(\psi)},  
    \end{align*} 
    with $\Delta(\tr_{X'R}(\mathcal{N}_i \otimes \cid)(\psi)) \leq D_i$ for $i=1,2$.   
    Then, consider the channel $\lambda\mathcal{N}_1 + (1-\lambda)\mathcal{N}_2$.  
    By convexity of $\Delta$, 
    \begin{equation}    
    \Delta \left( \; \tr_{X'R}((\lambda\mathcal{N}_1 + (1-\lambda)\mathcal{N}_2)\otimes \cid)(\psi) \; \right) \leq \lambda D_1 + (1-\lambda)D_2 \,.   
    \label{eq:distbdd} 
    \end{equation} 
    We have the following chain of inequalities 
    \begin{equation*}
        \begin{aligned}
            f(\lambda D_1 + (1-\lambda)D_2) 
            &\leq \frac{1}{2} I(B:XX'R)_{((\lambda\mathcal{N}_1 + (1-\lambda)\mathcal{N}_2)\otimes \cid)(\psi)} \\
            &\leq \lambda  \frac{1}{2} I(B:XX'R)_{(\mathcal{N}_1 \otimes \cid)(\psi)} 
             + (1-\lambda) \frac{1}{2} I(B:XX'R)_{(\mathcal{N}_2 \otimes \cid)(\psi)} \\ 
            &= \lambda f(D_1) + (1-\lambda) f(D_2),  
        \end{aligned}
    \end{equation*}
    where the first inequality comes from Eq.~(\ref{eq:distbdd}) and the minimizing 
    definition of $f(D)$, 
    and the second inequality follows from Lemma~\ref{lem:convexity}. 
\end{proof}

This lemma implies the continuity of $f(D)$.  
Using this lemma and Theorem \ref{thm:assisted}, 
$\R_{\rm ea}(D)$ is convex and continuous.  

\subsection{Properties of the unassisted rate-distortion function}~\label{sec:gd}
\begin{lemma}~\label{lem:gdcts}
   Consider the functions $g(D)$, $g_k(D)$ defined at the beginning of the proof of  
   Theorem \ref{thm:unassisted} and $D>0$.  Then, 
   (1) $g(D)$ is continuous, (2) $g_k(D)$ converges to $g(D)$ uniformly, 
   and (3) for sufficiently large $k$, $g_k(D)$ is continuous.  
\end{lemma} 
\begin{proof}

We first observe and sketch the derivations of some useful properties of 
  $g_k(D)$. Next, we show that for each $D$, $g_k(D)$ converges, so $g(D)$ is
  well-defined. Then we establish some additional properties of $g(D)$.

Since the distortion $\Delta$ is convex and bounded (see the discussion 
after Definition \ref{def:distortion}), the domain for each $g_k(D)$ is an 
interval $[0,D_{\rm max}]$.  Furthermore, each $g_k(D)$ is bounded, and 
decreasing with $D$, with $g_k(D) = 0$ in some interval $[D_k,D_{\max}]$ 
for some $D_k\geq 0$.
Consider a fixed $D$ and the sequence $\{g_k(D)\}_k$.  Note that the 
sequence may not be decreasing (with $k$).  Instead, we have the property that 
for any natural number $l$, $g_{lk}(D) \leq g_k(D)$; to
see this, let $\mathcal{N}^{A^kJ^k\to B^k}_k$ be a minimizing map for
$g_k$, then $\left( \mathcal{N}^{A^kJ^k\to B^k}_k \right)^{\otimes l}$
is a feasible solution for the minimization for $g_{lk}$ because both
types of local error criteria are preserved, and because the
entanglement of purification is subadditive.
Similarly, $g_{k + l}(D) \leq \frac{k}{k+l} g_k(D) + \frac{l}{k+l} g_l(D)$.   
In particular, if we compare $g_{lk+r}(D)$ with $g_k(D)$ 
for some $r \in \{1,2,\cdots,k-1\}$, $g_{lk+r}(D)$ cannot be larger
than $g_k(D)$ by a constant divided by $l$. 
The above properties imply that $\{g_k(D)\}_k$ converges.  To  
see this, first there are converging subsequences of $\{g_k(D)\}_k$ 
because the range is compact.  Second if the limsup is strictly
larger than the liminf, say, by an amount $\alpha$, we can take 
$k$ with $g_k(D)$ no more than $\alpha/100$ from the liminf, and 
take some very large $l$ so that $g_{kl+r}(D)$ is no less than 
$\alpha/100$ from the limsup, which gives a contradiction for 
very large $l$.  So, $g(D)$ is well-defined. 
For each $k$, from the definition of $g_k(D)$, $g_k$ is decreasing
with $D$ so $g(D)$ is also a decreasing function of $D$.

{\bf Proof of (1):} At the beginning of the appendix of
Ref.~\cite{Datta2013b}, the authors proved the convexity of their
rate distortion function $R^q(D)$ for pure state compression.
Their ideas for the convexity proof apply also to our expression
$g(D)$.
However, since we do not have a bound for the domain 
for the minimization for $g(D)$, the convexity cannot be used to show
continuity.

Here, we provide a direct proof for the continuity of $g(D)$, partly 
by adapting the ideas in Ref.~\cite{Datta2013b}.  
Fix $\epsilon >0$.  
Let $D_2$ be an interior point in $[0,D_{\rm max}]$, and we fix some
$D_1 \in [0,D_{\rm max}]$ so that $D_1 < D_2$.
For each $i=1,2$, let $k_i$ satisfy $0 \leq g_{k_i}(D_i) - g(D_i) \leq \epsilon$.
Let $k=k_1k_2$ so that
$0 \leq g_{k}(D_i) - g(D_i) \leq g_{k_i}(D_i) - g(D_i) \leq \epsilon$.  
To show that $g$ is continuous at $D_2$ from below,
let $k' = k(l+1)$ for some natural number $l$ to be chosen later,
$D' = (D_1 + l D_2)/(l+1)$, and $\cN_{i,k}$ be the optimal
channel for $g_k(D_i)$.  Then, on $k'$ copies,
the channel $\cN_{1,k} \otimes (\cN_{2,k})^{\otimes l}$ achieves
an \emph{average} local error criterion no more than $D'$,
and an entanglement of purification at most
$k g_k(D_1) + k l g_k(D_2)$, so,
$g_{k(l+1)}(D') \leq  (g_k(D_1) + l g_k(D_2))/(l+1)$.  
Putting the above together,
\begin{align*}
0 
& \leq g(D')-g(D_2) 
\\ & \leq g_{k(l+1)}(D') - g(D_2)  
\\ & \leq (g_k(D_1) + l g_k(D_2))/(l+1) - g(D_2)
\\ & = (g_k(D_1) - g_k(D_2))/(l+1) +  g_k(D_2) - g(D_2)
\\ & \leq (g_k(D_1) - g_k(D_2))/(l+1) + \epsilon 
\\ & \leq 2 \epsilon
\end{align*}
if we choose $l \geq \lceil (g_k(D_1) - g_k(D_2)) / \epsilon \rceil$.
Finally, by monotonicity of $g$, $\forall D \in (D',D_2)$,
$0 \leq g(D)-g(D_2) \leq 2 \epsilon$.

To show that $g$ is continuous at $D_2$ from above, keep the above notations and further 
define $D_3$ to satisfy $D_2 = (D'+D_3)/2$, and define $k_3$ so that 
$g_{k_3}(D') - g(D') \leq \epsilon$ and $g_{k_3}(D_3) - g(D_3) \leq \epsilon$.  
Then, $g(D_2) \leq g_{2k_3}(D_2) \leq ( g_{k_3}(D') +  g_{k_3}(D_3) ) / 2 
\leq ( g(D') +  g(D_3) ) / 2 + \epsilon \leq  ( g(D_2) +  g(D_3) ) / 2 + 2 \epsilon$.  
Rearranging the terms, $g(D_2) - g(D_3) \leq 4 \epsilon$, and by monotonicity, 
$g(D_2) - g(D) \leq 4 \epsilon$ for all $D \in [D_2,D_3]$.  

We do not know if $g(D)$ is continuous at $D=0$, but $g_k(D)$ 
and $g(D)$ are continuous at $D=D_{\max}$.  For blind compression with 
trivial side information, $g_k(D)$ and $g(D)$ are continuous at $D=0$. 

{\bf Proof of (2):} Fix any $\epsilon > 0$.  
Since $g(D)$ is decreasing and continuous, we can find 
$0 < D_1 \leq D_2 \leq \cdots \leq D_m$ so that 
$g(D_i) - g(D_{i+1}) \leq \epsilon$ and $g(D_m) = 0$.  
Let $k_i$ be such that $g_{k_i}(D_i) - g(D_i) \leq \epsilon$. 
Let $k = k_1 k_2 \cdots k_m$.  Then, 
$g_{k}(D_i) - g(D_i) \leq \epsilon$ for all $i$.
Then, for each $D \in [D_i,D_{i+1}]$, 
$$0 \leq g_k(D) - g(D) \leq g_k(D_i) - g(D_{i+1}) 
\leq g(D_i) + \epsilon - g(D_{i+1}) \leq 2\epsilon,$$
where the first inequality comes from the monotonicity of 
$g$ and $g_k$, and the second inequality comes from the 
choice of $k$, and the last inequality comes from the 
choice of the $D_i$'s.  

{\bf Proof of (3):} Fix any $\epsilon > 0$.  Using uniform convergence proved in part
(2), there exists a sufficiently large $k$ so that for all $D$,
$g_k(D) - g(D) \leq \epsilon$.  Let $D_2$ be any interior point
of $[0, D_{\rm max}]$.  By (1) let $D_1 < D_2 < D_3$ satisfy 
$g(D_{i})-g(D_{i{+}1}) \leq \epsilon$ for $i=1,2$.  
Then, $g_k(D_i) - g_k(D_{i{+}1}) \leq g(D_i) + \epsilon - g(D_{i{+}1}) \leq 2 \epsilon$. 
So, $g_k$ is continuous. 

\end{proof}

Together with Theorem \ref{thm:unassisted}, the above establishes that $\R(D)$
is convex and continuous in $D$ for $D>0$.

\subsection{Proof of Theorem \ref{thm:KI_redundant}}
\label{app:redundant-part-pf}

In the proof, we assume all definitions and notations defined in 
Theorem \ref{thm:KI_redundant}.  
We first prove the following lemma. 
\begin{lemma}~\label{lem:minimum_proj}
    Consider the setup stated in  Theorem \ref{thm:KI_redundant}.  
    Let $\mathcal{N}: CNQJ \to CNQ$ be a channel achieving the minimum of Eq.~\eqref{eq:KI_ea_rate} 
    for distortion $D$. 
    Let $\mathcal{P}^{C}: C \to C$ be a channel corresponding to the projective measurement on $C$ 
    with respect to the basis $\{\ket{c}^{C}\}_{c}$; that is, 
    $\mathcal{P}^{C}(\xi) = \sum_{c} \bra{c}\xi \ket{c} \proj{c}^{C}$.
    Then, $\mathcal{N}\circ \mathcal{P}^{C}$ 
    also achieves the minimum of Eq.~\eqref{eq:KI_ea_rate}. 
\end{lemma}
\begin{proof}
    We first prove the feasibility and then the optimality of 
    $\mathcal{N}\circ \mathcal{P}^{C}$ 
    in the minimization in Eq.~\eqref{eq:KI_ea_rate}.
    Feasibility means  
    that $\mathcal{N}\circ \mathcal{P}^{C}$ 
    satisfies the distortion condition $1 - F \left(\omega^{CNQX},  \left(\mathcal{N} \circ \mathcal{P}^{C}\right)\left(\omega^{CNQJX}\right) \right) \leq D$. 
    Observe from the form of $\omega$ in Eq.~\eqref{eq:redundant_source} that 
    \begin{equation}
        \left(\mathcal{P}^{C}\otimes \mathcal{I}^{NQJX}\right)\left(\omega^{CNQJX}\right) = \omega^{CNQJX} \,,
    \end{equation}
    therefore, 
    \begin{equation}
        \begin{aligned}
            &F \left(\omega^{CNQX},  \left(\mathcal{N}\circ \mathcal{P}^{C}\right)\left(\omega^{CNQJX}\right) \right) \\ 
            &~= ~ F \left(\omega^{CNQX},  \mathcal{N}\left(\omega^{CNQJX}\right) \right) \\ 
            &~\geq ~1 - D, 
        \end{aligned}
    \end{equation}
    where the inequality holds because $\mathcal{N}$ satisfies the distortion condition. 
    Hence, $\mathcal{N}\circ \mathcal{P}^{C}$ is feasible. 

    For the optimality of 
    $\mathcal{N}\circ \mathcal{P}^{C}$ 
    in the minimization in Eq.~\eqref{eq:KI_ea_rate}, 
    using the form of $\ket{\omega}$ in Eq.~\eqref{eq:redundant_source_purified}, 
    we have 
    $\mathcal{N} \circ \mathcal{P}^{C} (\proj{\omega}) 
    = \mathcal{P}^{R_C} \circ \mathcal{N} (\proj{\omega})$, where 
    $\mathcal{P}^{R_C}$ is the same channel as $\mathcal{P}^C$ now defined on system $R_C$.  Therefore, 
    \begin{equation}
        \begin{aligned}
            \frac{1}{2} I(CNQ:R_CR_NR_QXX')_{(\mathcal{N}\circ\mathcal{P}^{C})(\proj{\omega})} 
            &= \frac{1}{2} I(CNQ:R_CR_NR_QXX')_{(\mathcal{P}^{R_C}\circ \mathcal{N})(\proj{\omega})} 
        \end{aligned}
    \end{equation}
and the RHS is upper-bounded by $\frac{1}{2} I(CNQ:R_CR_NR_QXX')_{\mathcal{N}(\proj{\omega})}$ 
by the data processing inequality. 
    Since $\mathcal{N}$ achieves the minimum, 
    $\mathcal{N} \circ \mathcal{P}^{C}$ also achieves the minimum, 
    which shows the optimality. 
\end{proof}

We now proceed to prove Theorem~\ref{thm:KI_redundant}. 
\begin{proof}[Proof of Theorem~\ref{thm:KI_redundant}]
    Recall that we may remove and attach redundant system $N$ from $\omega^{CNQX}$ 
    by using quantum channels $\mathcal{K}_{\mathrm{off}}$ and $\mathcal{K}_{\mathrm{on}}$. 

    We first show   
    \begin{equation}~\label{eq:redundant_achievability}
        \min_{\mathcal{N}^{CNQJ \to CNQ}} \frac{1}{2} I(CNQ: R_CR_NR_QXX')_{\mathcal{N}(\proj{\omega})}
        \leq 
        \min_{\mathcal{M}^{CQJ \to CQ}} \frac{1}{2} I(CQ: R_CR_QXX')_{\mathcal{M}(\proj{\tilde{\omega}})} \,, 
    \end{equation}
    where the minimum on the LHS is subject to the constraint    
    $1 - F \left(\omega^{CNQX},  \mathcal{N}\left(\omega^{CNQJX}\right) \right) \leq D$, 
    and the minimum on the RHS is subject to the constraint
    $1 - F \left(\omega^{CQX},  \mathcal{M}\left(\omega^{CQJX}\right) \right) \leq D$. 
    Let $\mathcal{M}$ be a quantum channel achieving the minimization on the RHS. 
    By Lemma~\ref{lem:minimum_proj}, $\mathcal{M}\circ \mathcal{P}^{C}$ also achieves the minimum.   
    To show inequality \eqref{eq:redundant_achievability},
    we will show that 
    \begin{enumerate}
        \item $\mathcal{K}_{\mathrm{on}} \circ \mathcal{M} \circ \mathcal{K}_{\mathrm{off}}$ 
        is feasible for the LHS, and 
        \item $\frac{1}{2} I(CNQ: R_CR_NR_QXX')_{(\mathcal{K}_{\mathrm{on}} \circ \mathcal{M} \circ \mathcal{K}_{\mathrm{off}})(\proj{\omega})} 
        \leq \frac{1}{2} I(CQ: R_CR_QXX')_{(\mathcal{M}\circ\mathcal{P}^{C})(\proj{\tilde{\omega}})}$. 
    \end{enumerate}
    The feasibility in step 1 follows from  
    \begin{equation}
        \begin{aligned}
            F(\omega^{CNQX},(\mathcal{K}_{\mathrm{on}} \circ \mathcal{M} \circ \mathcal{K}_{\mathrm{off}})(\omega^{CNQX})) 
            &= F(\omega^{CNQX},(\mathcal{K}_{\mathrm{on}} \circ \mathcal{M})(\omega^{CQX})) \\ 
            &= F(\mathcal{K}_{\mathrm{on}}(\omega^{CQX}),(\mathcal{K}_{\mathrm{on}} \circ \mathcal{M})(\omega^{CQX})) \\
            &\geq F(\omega^{CQX},\mathcal{M}(\omega^{CQX})) \\ 
            &\geq 1 - D. 
        \end{aligned}
    \end{equation}
    In the above, the first and second lines follow from $\mathcal{K}_{\mathrm{off}}(\omega^{CNQX}) = \omega^{CQX}$ 
    and $\mathcal{K}_{\mathrm{on}}(\omega^{CQX}) = \omega^{CNQX}$ respectively. 
    The third line is due to the monotonicity of fidelity, and the last line is due to  
    the feasibility of $\mathcal{M}$ for the distortion $D$. 
    For the inequality 
    $\frac{1}{2} I(CNQ: R_CR_NR_QXX')_{(\mathcal{K}_{\mathrm{on}} \circ \mathcal{M} \circ
    \mathcal{K}_{\mathrm{off}})(\proj{\omega})} 
    \leq \frac{1}{2} I(CQ: R_CR_QXX')_{(\mathcal{M}\circ \mathcal{P}^{C})(\proj{\tilde{\omega}})}$
    in step 2, we start with the LHS and 
    apply the data processing inequality to obtain  
    \begin{equation}~\label{eq:ineq_1}
        \frac{1}{2} I(CNQ: R_CR_NR_QXX')_{(\mathcal{K}_{\mathrm{on}} \circ \mathcal{M} \circ \mathcal{K}_{\mathrm{off}})(\proj{\omega})} 
        \leq \frac{1}{2} I(CQ: R_CR_NR_QXX')_{(\mathcal{M} \circ \mathcal{K}_{\mathrm{off}})(\proj{\omega})}. 
    \end{equation}
    Observe that 
    \begin{equation}
        \mathcal{K}_{\mathrm{off}}(\proj{\omega}) = \sum_{x,x'} \sum_{c} \sqrt{p_{x}p_{x'} p_{c|x}p_{c|x'}} \proj{c}^{C} \otimes \proj{c}^{R_C} 
        \otimes \omega_c^{R_N} \otimes \ketbra{\rho_{cx}}{\rho_{cx'}}^{QR_Q} 
        \otimes \ketbra{j_x}{j_{x'}}^{J} \otimes \ketbra{x}{x'}^X \otimes \ketbra{x}{x'}^{X'}. 
    \end{equation}
    By applying $\mathcal{M}\circ \mathcal{K}_{\mathrm{off}}$ to $\proj{\omega}$, 
    the resulting state is 
    \begin{equation}~\label{eq:resulting_state_1}
        \sum_{x,x'} \sum_{c} \sqrt{p_{x}p_{x'} p_{c|x}p_{c|x'}} \mathcal{M} \left(\proj{c}^{C} 
         \otimes \ketbra{\rho_{cx}}{\rho_{cx'}}^{QR_Q} 
        \otimes \ketbra{j_x}{j_{x'}}^{J}\right) \otimes \proj{c}^{R_C} \otimes \omega_c^{R_N} \otimes \ketbra{x}{x'}^X \otimes \ketbra{x}{x'}^{X'}. 
    \end{equation}
    In Eq.~\eqref{eq:resulting_state_1}, the state on $CQR_NR_QXX'$ is a product state
    between $CQR_QXX'$ and $R_N$ when conditioned on $R_C$; 
    in particular, for each $c$, the state on $CQR_NR_QXX'$ is 
    \begin{equation}
        \left(\sum_{x,x'} \sqrt{p_{x}p_{x'} p_{c|x}p_{c|x'}} \mathcal{M} \left(\proj{c}^{C} 
         \otimes \ketbra{\rho_{cx}}{\rho_{cx'}}^{QR_Q} 
        \otimes \ketbra{j_x}{j_{x'}}^{J}\right)\otimes \ketbra{x}{x'}^X \otimes \ketbra{x}{x'}^{X'}\right) \otimes \omega_c^{R_N}. 
    \end{equation}
    Hence,
    \begin{equation}~\label{eq:ineq_2}
        I(CQ: R_CR_NR_QXX')_{(\mathcal{M} \circ \mathcal{K}_{\mathrm{off}})(\proj{\omega})} = I(CQ: R_CR_QXX')_{(\mathcal{M} \circ \mathcal{K}_{\mathrm{off}})(\proj{\omega})}. 
    \end{equation}
    From Eq.~\eqref{eq:resulting_state_1}, when we trace out $R_N$ from $(\mathcal{M}\circ \mathcal{K}_{\mathrm{off}})(\proj{\omega})$, 
    the resulting state is 
    \begin{equation}~\label{eq:resulting_state_2}
        \sum_{x,x'} \sum_{c} \sqrt{p_{x}p_{x'} p_{c|x}p_{c|x'}} \mathcal{M} \left(\proj{c}^{C} 
         \otimes \ketbra{\rho_{cx}}{\rho_{cx'}}^{QR_Q} 
        \otimes \ketbra{j_x}{j_{x'}}^{J}\right) \otimes \proj{c}^{R_C} \otimes \ketbra{x}{x'}^X \otimes \ketbra{x}{x'}^{X'}, 
    \end{equation}
    which is the same state as $(\mathcal{M}\circ \mathcal{P}^{C})(\proj{\tilde{\omega}})$. 
    Therefore, $\tr_{R_N}\left[(\mathcal{M}\circ \mathcal{K}_{\mathrm{off}})(\proj{\omega})\right] = (\mathcal{M}\circ \mathcal{P}^{C})(\proj{\tilde{\omega}})$, 
    and thus 
    \begin{equation}~\label{eq:ineq_3}
        I(CQ: R_CR_QXX')_{(\mathcal{M} \circ \mathcal{K}_{\mathrm{off}})(\proj{\omega})} = I(CQ: R_CR_QXX')_{(\mathcal{M}\circ\mathcal{P}^{C})(\proj{\tilde{\omega}})}. 
    \end{equation}
    From Eqs.~\eqref{eq:ineq_1}, \eqref{eq:ineq_2}, and \eqref{eq:ineq_3}, 
    \begin{equation}
        \frac{1}{2} I(CNQ: R_CR_NR_QXX')_{(\mathcal{K}_{\mathrm{on}} \circ \mathcal{M} \circ \mathcal{K}_{\mathrm{off}})(\proj{\omega})} 
    \leq \frac{1}{2} I(CQ: R_CR_QXX')_{(\mathcal{M}\circ \mathcal{P}^{C})(\proj{\tilde{\omega}})}, 
    \end{equation} 
    as claimed in step 2.

    Next, we show the inequality opposite to \eqref{eq:redundant_achievability}; that is, 
    \begin{equation}~\label{eq:redundant_converse}
        \min_{\mathcal{N}^{CNQJ \to CNQ}} \frac{1}{2} I(CNQ: R_CR_NR_QXX')_{\mathcal{N}(\proj{\omega})}
        \geq 
        \min_{\mathcal{M}^{CQJ \to CQ}} \frac{1}{2} I(CQ: R_CR_QXX')_{\mathcal{M}(\proj{\tilde{\omega}})} , 
    \end{equation}
    where the minimum on the LHS is subject to the constraint    
    $1 - F \left(\omega^{CNQX},  \mathcal{N}\left(\omega^{CNQJX}\right) \right) \leq D$, 
    and the minimum on the RHS is subject to the constraint
    $1 - F \left(\omega^{CQX},  \mathcal{M}\left(\omega^{CQJX}\right) \right) \leq D$. 
    Let $\mathcal{N}$ be a quantum channel achieving the minimum of LHS. 
    By Lemma~\ref{lem:minimum_proj}, $\mathcal{N}\circ \mathcal{P}^{C}$ also achieves the minimum.   
    To show Eq.~\eqref{eq:redundant_converse}, we will show that 
    \begin{enumerate}
        \item $\mathcal{K}_{\mathrm{off}} \circ \mathcal{N} \circ \mathcal{K}_{\mathrm{on}}$ 
        is feasible for RHS, and 
        \item $\frac{1}{2} I(CNQ: R_CR_NR_QXX')_{(\mathcal{N}\circ\mathcal{P}^{C})(\proj{\omega})} 
        \geq \frac{1}{2} I(CQ: R_CR_QXX')_{(\mathcal{K}_{\mathrm{off}} \circ \mathcal{N} \circ \mathcal{K}_{\mathrm{on}})(\proj{\tilde{\omega}})}$. 
    \end{enumerate}
    The feasibility can be shown as follows. 
    \begin{equation}
        \begin{aligned}
            F(\omega^{CQX},(\mathcal{K}_{\mathrm{off}} \circ \mathcal{N} \circ \mathcal{K}_{\mathrm{on}})(\omega^{CQX})) 
            &= F(\omega^{CQX},(\mathcal{K}_{\mathrm{off}} \circ \mathcal{N})(\omega^{CNQX})) \\ 
            &= F(\mathcal{K}_{\mathrm{off}}(\omega^{CNQX}),(\mathcal{K}_{\mathrm{off}} \circ \mathcal{N})(\omega^{CNQX})) \\
            &\geq F(\omega^{CNQX},\mathcal{N}(\omega^{CNQX})) \\ 
            &\geq 1 - D. 
        \end{aligned}
    \end{equation}
    The first line follows because $\mathcal{K}_{\mathrm{on}}(\omega^{CQX}) = \omega^{CNQX}$. 
    The second line follows because  $\mathcal{K}_{\mathrm{off}}(\omega^{CNQX}) = \omega^{CQX}$. 
    The third line is monotonicity of fidelity. 
    The fourth line follows because $\mathcal{N}$ is feasible for distortion $D$. 
    We then show the inequality in the second step: $\frac{1}{2} I(CNQ: R_CR_NR_QXX')_{(\mathcal{N}\circ\mathcal{P}^{C})(\proj{\omega})} 
    \geq \frac{1}{2} I(CQ: R_CR_QXX')_{(\mathcal{K}_{\mathrm{off}} \circ \mathcal{N} \circ \mathcal{K}_{\mathrm{on}})(\proj{\tilde{\omega}})}$.
    By data processing inequality, 
    \begin{equation}~\label{eq:ineq_4}
        \frac{1}{2} I(CQ: R_CR_QXX')_{(\mathcal{K}_{\mathrm{off}} \circ \mathcal{N} \circ \mathcal{K}_{\mathrm{on}})(\proj{\tilde{\omega}})} 
        \leq \frac{1}{2} I(CNQ: R_CR_QXX')_{( \mathcal{N} \circ \mathcal{K}_{\mathrm{on}})(\proj{\tilde{\omega}})}. 
    \end{equation}
    Observe that 
    \begin{equation}
        \mathcal{K}_{\mathrm{on}}(\proj{\tilde{\omega}}) = \sum_{x,x'} \sum_{c} \sqrt{p_{x}p_{x'} p_{c|x}p_{c|x'}} \proj{c}^{C} \otimes \proj{c}^{R_C} 
        \otimes \omega_c^{N} \otimes \ketbra{\rho_{cx}}{\rho_{cx'}}^{QR_Q} 
        \otimes \ketbra{j_x}{j_{x'}}^{J} \otimes \ketbra{x}{x'}^X \otimes \ketbra{x}{x'}^{X'}. 
    \end{equation}
    By applying $\mathcal{N}\circ \mathcal{K}_{\mathrm{on}}$ to $\proj{\tilde{\omega}}$, 
    the resulting state is 
    \begin{equation}~\label{eq:resulting_state_3}
        \sum_{x,x'} \sum_{c} \sqrt{p_{x}p_{x'} p_{c|x}p_{c|x'}} \mathcal{N} \left(\proj{c}^{C} 
         \otimes \omega_c^{N} \otimes \ketbra{\rho_{cx}}{\rho_{cx'}}^{QR_Q} 
        \otimes \ketbra{j_x}{j_{x'}}^{J}\right) \otimes \proj{c}^{R_C} \otimes \ketbra{x}{x'}^X \otimes \ketbra{x}{x'}^{X'}. 
    \end{equation}
    This state can also be obtained by tracing out $R_N$ from 
    \begin{equation}~\label{eq:resulting_state_4}
        \sum_{x,x'} \sum_{c} \sqrt{p_{x}p_{x'} p_{c|x}p_{c|x'}} \mathcal{N} \left(\proj{c}^{C} 
         \otimes \proj{\omega_c}^{NR_N} \otimes \ketbra{\rho_{cx}}{\rho_{cx'}}^{QR_Q} 
        \otimes \ketbra{j_x}{j_{x'}}^{J}\right) \otimes \proj{c}^{R_C} \otimes \ketbra{x}{x'}^X \otimes \ketbra{x}{x'}^{X'}. 
    \end{equation}
    Note that $\mathcal{N}$ does not act on $R_N$. 
    The state in Eq.~\eqref{eq:resulting_state_4} is obtained by applying $\mathcal{N}\circ \mathcal{P}^{C}$ to $\proj{\omega}$. 
    Therefore, 
    $\tr_{R_N}\left[(\mathcal{N}\circ \mathcal{P}^{C})(\proj{\omega})\right] = (\mathcal{N}\circ \mathcal{K}_{\mathrm{on}})(\proj{\tilde{\omega}})$, 
    and by data processing inequality, 
    \begin{equation}~\label{eq:ineq_5}
        I(CNQ: R_C R_QXX')_{( \mathcal{N} \circ \mathcal{K}_{\mathrm{on}})(\proj{\tilde{\omega}})} \leq  I(CNQ: R_CR_NR_QXX')_{(\mathcal{N}\circ\mathcal{P}^{C})(\proj{\omega})}. 
    \end{equation}
    From Eqs.~\eqref{eq:ineq_4} and \eqref{eq:ineq_5}, 
    \begin{equation}
        \frac{1}{2} I(CQ: R_CR_QXX')_{(\mathcal{K}_{\mathrm{off}} \circ \mathcal{N} \circ \mathcal{K}_{\mathrm{on}})(\proj{\tilde{\omega}})} \leq 
        \frac{1}{2} I(CNQ: R_CR_NR_QXX')_{(\mathcal{N}\circ\mathcal{P}^{C})(\proj{\omega})}, 
    \end{equation} 
    which proves the inequality in step 2. 
\end{proof}

\subsection{The proof for Theorem \ref{thm:KI_redundant_un}} \label{app:redundant_un}

    We retain all the notations and definitions leading up to Theorem \ref{thm:KI_redundant_un}.  
    We will use the following lemma, which follows similarly as in Lemma~\ref{lem:minimum_proj}. 
    \begin{lemma}~\label{lem:minimum_proj_un}
        Consider the setup described in Theorem \ref{thm:KI_redundant_un}. 
        Let $k$ be a fixed positive integer, and
        let $\mathcal{N}_k: C^kN^kQ^kJ^k \to C^kN^kQ^k$ be a channel achieving the minimum of 
        \begin{equation}~\label{eq:min_fixed_k}
            \min_{\mathcal{N}_k} E_p(C^kN^kQ^k: R_C^kR_N^kR_Q^kX^kX'^k)_{\mathcal{N}_{k}\left(\proj{\omega}^{\otimes k}\right)}
        \end{equation} 
        where the minimum is taken over all quantum channels $\mathcal{N}_k: C^kN^kQ^kJ^k \to C^kN^kQ^k$ with 
        \begin{equation}
            1 - \frac{1}{k}\sum_{j=1}^{k} F \left(\omega^{C_jN_jQ_jX_j}, \left(\mathcal{N}_{k}\left(\proj{\omega}^{\otimes k}\right)\right)^{C_jN_jQ_jX_j} \right) \leq D. 
        \end{equation}
        Let $\mathcal{P}^{C}: C \to C$ be the projective measurement on $C$ 
        as defined in Lemma~\ref{lem:minimum_proj}. 
        Then, $\mathcal{N}_k\circ (\mathcal{P}^{C})^{\otimes k}$ 
        also achieves the minimum of Eq.~\eqref{eq:min_fixed_k}. 
    \end{lemma}

    \begin{proof}[Proof Sketch]
    We first show the feasibility and then show the optimality. 
    Similar to the proof of Lemma~\ref{lem:minimum_proj}, 
    the feasibility of $\mathcal{N}_k\circ (\mathcal{P}^{C})^{\otimes k}$
    follows from 
    \begin{equation}
        (\mathcal{P}^{C}\otimes \mathcal{I}^{NQJX})(\omega^{CNQJX}) = \omega^{CNQJX}
    \end{equation} 
    and the feasibility of $\mathcal{N}_k$.
    The optimality also follows similarly because 
    \begin{equation}
        (\mathcal{N}_k\circ (\mathcal{P}^{C})^{\otimes k})(\proj{\omega}^{\otimes k}) 
        = ((\mathcal{P}^{R_C})^{\otimes k} \circ \mathcal{N}_k)(\proj{\omega}^{\otimes k}) 
    \end{equation}
    with the measurement $\mathcal{P}^{R_C}$ acting on $R_C$, as defined in
    Lemma~\ref{lem:minimum_proj}.  
    Indeed, we have 
    \begin{equation}
        \begin{aligned}
            E_p(C^kN^kQ^k:R^k_CR^k_NR^k_QX^kX'^k)_{(\mathcal{N}_k\circ (\mathcal{P}^{C})^{\otimes k})(\proj{\omega}^{\otimes k})} 
            &= E_p(C^kN^kQ^k:R^k_CR^k_NR^k_QX^kX'^k)_{((\mathcal{P}^{R_C})^{\otimes k} \circ \mathcal{N}_k)(\proj{\omega}^{\otimes k})} \\ 
            &\leq  E_p(C^kN^kQ^k:R^k_CR^k_NR^k_QX^kX'^k)_{\mathcal{N}_k(\proj{\omega}^{\otimes k})}, 
        \end{aligned}
    \end{equation}
    where the inequality follows from monotonicity of entanglement of purification. 
    \end{proof}

    \begin{proof}[Proof Sketch of Theorem~\ref{thm:KI_redundant_un}]
        It suffices to show that for each $k$, 
        \begin{equation}~\label{eq:redundant_achievability_un}
            \min_{\mathcal{N}_k} E_p(C^kN^kQ^k: R_C^kR_N^kR_Q^kX^kX'^k)_{\mathcal{N}_k(\proj{\omega}^{\otimes k})}
            =
            \min_{\mathcal{M}_k} E_p(C^kQ^k: R_C^kR_Q^kX^kX'^k)_{\mathcal{M}_k(\proj{\tilde{\omega}}^{\otimes k})} ,
        \end{equation}
        where the minimum on the LHS is taken over all quantum channels $\mathcal{N}_k: C^kN^kQ^kJ^k \to C^kN^kQ^k$ 
        with 
        \begin{equation}
            1 - \frac{1}{k}\sum_{j=1}^{k} F \left(\omega^{C_jN_jQ_jX_j}, \left(\mathcal{N}_{k}\left(\proj{\omega}^{\otimes k}\right)\right)^{C_jN_jQ_jX_j} \right) \leq D,  
        \end{equation}
        and the minimum on the RHS is taken over all quantum channels $\mathcal{M}_k: C^kQ^kJ^k \to C^kQ^k$ 
        with 
        \begin{equation}
            1 - \frac{1}{k}\sum_{j=1}^{k} F \left(\omega^{C_jQ_jX_j}, \left(\mathcal{M}_{k}\left(\proj{\tilde{\omega}}^{\otimes k}\right)\right)^{C_jQ_jX_j} \right) \leq D. 
        \end{equation}
        We first show ``$\leq$'' in the above, starting from the RHS, with $\mathcal{M}_k$ achieving the minimum.  
        By Lemma~\ref{lem:minimum_proj_un}, $\mathcal{M}_k\circ (\mathcal{P}^{C})^{\otimes k}$ also achieves the minimum.   
        We will show that 
        \begin{enumerate}
            \item $\mathcal{K}_{\mathrm{on}}^{\otimes k} \circ \mathcal{M}_k \circ \mathcal{K}_{\mathrm{off}}^{\otimes k}$ 
            is feasible for the LHS, and 
            \item $E_p(C^kN^kQ^k: R_C^kR_N^kR_Q^kX^kX'^k)_{(\mathcal{K}_{\mathrm{on}}^{\otimes k} \circ \mathcal{M}_k \circ \mathcal{K}_{\mathrm{off}}^{\otimes k})(\proj{\omega}^{\otimes k})} 
            \leq E_p(C^kQ^k: R_C^kR_Q^kX^kX'^k)_{(\mathcal{M}_k\circ(\mathcal{P}^{C})^{\otimes k})(\proj{\tilde{\omega}}^{\otimes k})}$. 
        \end{enumerate}
        The feasibility in the first step can be shown by using the relation
        \begin{equation}
            \mathcal{K}_{\mathrm{off}}(\omega^{CNQX}) = \omega^{CQX},\,\,\,\, 
            \mathcal{K}_{\mathrm{on}}(\omega^{CQX}) = \omega^{CNQX}, 
        \end{equation}
        monotonicity of fidelity, 
        and feasibility of $\mathcal{M}_k$. 
        (Note that the distortion condition is now given with respect to the average-case local error criterion. 
        In the proof for feasibility, we use the fact that for a quantum state $\rho^{A_1A_2}$ on joint system $A_1,A_2$ and a quantum channel $\mathcal{N}:A_2\to A_3$, 
        $\tr_{A_2}[\rho^{A_1A_2}] = \tr_{A_3}[\mathcal{N}(\rho^{A_1A_2})]$.)
        For the second step, for simplicity, we consider $k=1$; that is, we will show 
        $E_p(CNQ: R_CR_NR_QXX')_{(\mathcal{K}_{\mathrm{on}} \circ \mathcal{M}_1 \circ \mathcal{K}_{\mathrm{off}})(\proj{\omega})} 
        \leq E_p(CQ: R_CR_QXX')_{(\mathcal{M}_1\circ\mathcal{P}^{C})(\proj{\tilde{\omega}})}$. 
        The proof for general $k$ is similar. 
        By monotonicity of entanglement of purification, 
        \begin{equation}~\label{eq:ineq_un_1}
            E_p(CNQ: R_CR_NR_QXX')_{(\mathcal{K}_{\mathrm{on}} \circ \mathcal{M}_1 \circ \mathcal{K}_{\mathrm{off}})(\proj{\omega})}
            \leq E_p(CQ: R_CR_NR_QXX')_{(\mathcal{M}_1 \circ \mathcal{K}_{\mathrm{off}})(\proj{\omega})}. 
        \end{equation}
        We would like to show 
        \begin{equation}~\label{eq:achieve_subgoal}
            E_p(CQ: R_CR_NR_QXX')_{(\mathcal{M}_1 \circ \mathcal{K}_{\mathrm{off}})(\proj{\omega})} 
            \leq E_p(CQ: R_CR_QXX')_{(\mathcal{M}_1\circ\mathcal{P}^{C})(\proj{\tilde{\omega}})}. 
        \end{equation}
        The state $(\mathcal{M}_1\circ\mathcal{P}^{C})(\proj{\tilde{\omega}})$ can be purified as 
        \begin{equation}~\label{eq:purification_tilde_omega}
            |\tilde{\nu}\rangle \coloneqq \sum_{x}\sqrt{p_x}\sum_{c} \sqrt{p_{c|x}} \ket{\mu_{cx}}^{CQR_QW} \ket{c}^{R_C}\ket{c}^{R_C'}\ket{x}^{X}\ket{x}^{X'}, 
        \end{equation}
        where 
        \begin{equation}
            \ket{\mu_{cx}}^{CQR_QW} = U_{\mathcal{M}_1}(\ket{c}^{C}\ket{\rho_{cx}}^{QR_Q}\ket{j_x}^{J})
        \end{equation}
        with Stinespring dilation $U_{\mathcal{M}_1}: CQJ \hookrightarrow CQW$ of $\mathcal{M}_1$. 
        From Theorem~\ref{thm:entanglement_of_purification}, 
        recall that 
        \begin{equation}
            E_p(CQ: R_CR_QXX')_{(\mathcal{M}_1\circ\mathcal{P}^{C})(\proj{\tilde{\omega}})}
            = \min_{\tilde{\Lambda}:WR'_C\to F} S(CQF)_{\tilde{\Lambda}(\proj{\tilde{\nu}})}. 
        \end{equation}
        Let $\tilde{\Lambda}: WR'_C \to F$ be a quantum channel achieving the minimum in the RHS above; that is, 
        \begin{equation}~\label{eq:ineq_un_2}
            E_p(CQ: R_CR_QXX')_{(\mathcal{M}_1\circ\mathcal{P}^{C})(\proj{\tilde{\omega}})} 
            = S(CQF)_{\tilde{\Lambda}(\proj{\tilde{\nu}})}. 
        \end{equation}
        Here, from Eq.~\eqref{eq:purification_tilde_omega}, 
        the marginal state of $\tilde{\Lambda}(\proj{\tilde{\nu}})$ on system $CQF$ is 
        \begin{equation}~\label{eq:marginal_tilde_lambda_nu}
            \tilde{\Lambda}(\proj{\tilde{\nu}})^{CQF} = \sum_{x}\sum_{c} p_xp_{c|x} \tilde{\Lambda}(\mu_{cx}^{CQW} \otimes \proj{c}^{R_C'}). 
        \end{equation}
        On the other hand, the state $(\mathcal{M}_1 \circ \mathcal{K}_{\mathrm{off}})(\proj{\omega})$ can be purified as 
        \begin{equation}
            \ket{\nu} \coloneqq \sum_{x}\sqrt{p_x}\sum_{c} \sqrt{p_{c|x}} \ket{\mu_{cx}}^{CQR_QW} \ket{c}^{R_C}\ket{c}^{R_C'}\ket{\omega_c}^{NR_N}\ket{x}^{X}\ket{x}^{X'}. 
        \end{equation}
        Recall also that 
        \begin{equation}
            E_p(CQ: R_CR_NR_QXX')_{(\mathcal{M}_1 \circ \mathcal{K}_{\mathrm{off}})(\proj{\omega})} 
            = \min_{\Lambda:NWR'_C\to F} S(CQF)_{\Lambda(\proj{\nu})}. 
        \end{equation}
        By choosing $(\tilde{\Lambda}\circ \Tr_{N})$ as a feasible channel for the minimization in the RHS, 
        \begin{equation}~\label{eq:ineq_un_3}
            \begin{aligned}
                E_p(CQ: R_CR_NR_QXX')_{(\mathcal{M}_1 \circ \mathcal{K}_{\mathrm{off}})(\proj{\omega})} 
                \leq S(CQF)_{(\tilde{\Lambda}\circ \Tr_{N})(\proj{\nu})}. 
            \end{aligned}
        \end{equation}
        Observing that the marginal state of $(\tilde{\Lambda}\circ \Tr_{N})(\proj{\nu})$ on system $CQF$ coincides with $ \tilde{\Lambda}(\proj{\tilde{\nu}})^{CQF}$ in Eq.~\eqref{eq:marginal_tilde_lambda_nu}, 
        \begin{equation}~\label{eq:ineq_un_4}
            S(CQF)_{(\tilde{\Lambda}\circ \Tr_{N})(\proj{\nu})} = S(CQF)_{\tilde{\Lambda}(\proj{\tilde{\nu}})}. 
        \end{equation} 
        By Eqs.~\eqref{eq:ineq_un_2}, \eqref{eq:ineq_un_3}, and \eqref{eq:ineq_un_4}, 
        we have Eq.~\eqref{eq:achieve_subgoal}. 

        Now, we show the opposite inequality ``$\geq$'' in \eqref{eq:redundant_achievability_un} with all the associated constraints.  
        Let $\mathcal{N}_k$ be a quantum channel achieving the minimum of LHS. 
        By Lemma~\ref{lem:minimum_proj_un}, $\mathcal{N}_k\circ (\mathcal{P}^{C})^{\otimes k}$ also achieves the minimum.   
        We will show that 
        \begin{enumerate}
            \item $(\mathcal{K}_{\mathrm{off}})^{\otimes k} \circ \mathcal{N}_k \circ (\mathcal{K}_{\mathrm{on}})^{\otimes k}$ 
            is feasible for RHS, and 
            \item $E_p(C^kQ^k: R_C^kR_Q^kX^kX'^k)_{(\mathcal{K}_{\mathrm{off}}^{\otimes k} \circ \mathcal{N}_k \circ \mathcal{K}_{\mathrm{on}}^{\otimes k})(\proj{\tilde{\omega}}^{\otimes k})} 
            \leq E_p(C^kN^kQ^k: R_C^kR_N^kR_Q^kX^kX'^k)_{(\mathcal{N}_k\circ(\mathcal{P}^{C})^{\otimes k})(\proj{\omega}^{\otimes k})}$. 
        \end{enumerate}
        The feasibility in step 1 can be shown by using the relation
        \begin{equation}
            \mathcal{K}_{\mathrm{off}}(\omega^{CNQX}) = \omega^{CQX},\,\,\,\, 
            \mathcal{K}_{\mathrm{on}}(\omega^{CQX}) = \omega^{CNQX}, 
        \end{equation}
        the monotonicity of fidelity, 
        and the feasibility of $\mathcal{N}_k$. 
        For the inequality in the second step, 
        we apply the monotonicity of the entanglement of purification to obtain 
        \begin{equation}~\label{eq:ineq_un_5}
            E_p(C^kQ^k: R_C^kR_Q^kX^kX'^k)_{(\mathcal{K}_{\mathrm{off}}^{\otimes k} \circ \mathcal{N}_k \circ \mathcal{K}_{\mathrm{on}}^{\otimes k})(\proj{\tilde{\omega}}^{\otimes k})}
            \leq E_p(C^kN^kQ^k: R_C^kR_Q^kX^kX'^k)_{(\mathcal{N}_k \circ \mathcal{K}_{\mathrm{on}}^{\otimes k})(\proj{\tilde{\omega}}^{\otimes k})}. 
        \end{equation}
        Observe that 
        \begin{equation}
            \tr_{R_N^k}\left[(\mathcal{N}_k\circ(\mathcal{P}^{C})^{\otimes k})(\proj{\omega}^{\otimes k})\right] 
            = (\mathcal{N}_k \circ \mathcal{K}_{\mathrm{on}}^{\otimes k})(\proj{\tilde{\omega}}^{\otimes k}).  
        \end{equation}
        Hence, by the monotonicity of the entanglement of purification, 
        \begin{equation}~\label{eq:ineq_un_6}
            E_p(C^kN^kQ^k: R_C^kR_Q^kX^kX'^k)_{(\mathcal{N}_k \circ \mathcal{K}_{\mathrm{on}}^{\otimes k})(\proj{\tilde{\omega}}^{\otimes k})}
            \leq E_p(C^kN^kQ^k: R_C^kR_N^kR_Q^kX^kX'^k)_{(\mathcal{N}_k\circ(\mathcal{P}^{C})^{\otimes k})(\proj{\omega}^{\otimes k})}. 
        \end{equation}
        From \eqref{eq:ineq_un_5} and \eqref{eq:ineq_un_6}, we obtain
        $$E_p(C^kQ^k: R_C^kR_Q^kX^kX'^k)_{(\mathcal{K}_{\mathrm{off}}^{\otimes k} \circ \mathcal{N}_k \circ \mathcal{K}_{\mathrm{on}}^{\otimes k})(\proj{\tilde{\omega}}^{\otimes k})} 
            \leq E_p(C^kN^kQ^k: R_C^kR_N^kR_Q^kX^kX'^k)_{(\mathcal{N}_k\circ(\mathcal{P}^{C})^{\otimes k})(\proj{\omega}^{\otimes k})},$$ 
        which completes the proof. 
    \end{proof}

\end{document}